\documentclass[11pt]{article}
\usepackage{amsmath, amssymb, graphicx, dsfont, xcolor, hyperref, amsthm,subcaption}
\usepackage[margin=1in]{geometry}
\usepackage[hyphenbreaks]{breakurl}

\newtheorem{theorem}{Theorem}

\newtheorem{definition}{Definition}
\newtheorem{lemma}{Lemma}

\usepackage[square,sort,comma,numbers]{natbib}
\begin{document}


\title{Covariate-Adjusted Inference for Differential Analysis of High-Dimensional Networks}
\author{Aaron Hudson and
Ali Shojaie \\
Department of Biostatistics, University of Washington, Seattle, Washington}
\date{}

\maketitle

\begin{abstract}
Differences between biological networks corresponding to disease conditions can help delineate the underlying disease mechanisms.
Existing methods for differential network analysis do not account for dependence of networks on covariates.
As a result, these approaches may detect spurious differential connections induced by the effect of the covariates on both the disease condition and the network.
To address this issue, we propose a general covariate-adjusted test for differential network analysis.
Our method assesses differential network connectivity by testing the null hypothesis that the network is the same for individuals who have identical covariates and only differ in disease status.
We show empirically in a simulation study that the covariate-adjusted test exhibits improved type-I error control compared with na\"ive hypothesis testing procedures that do not account for covariates.
We additionally show that there are settings in which our proposed methodology provides improved power to detect differential connections.
We illustrate our method by applying it to detect differences in breast cancer gene co-expression networks by subtype.
\end{abstract}


\section{Introduction}


Complex diseases are often associated with aberrations in biological networks, such as gene regulatory networks and brain functional or structural connectivity networks \citep{barabasi2011network}.
Performing differential network analysis, or identifying connections in biological networks that change with disease condition, can provide insights into the disease mechanisms and lead to the identification of network-based biomarkers \citep{ideker2012differential, de2010differential}.

Probabilistic graphical models are commonly used to summarize the conditional independence structure of a set of nodes in a biological network.
A common approach to differential network analysis is to first estimate the graph corresponding to each disease condition and then assess between-condition differences in the graph.
For instance, when using Gaussian graphical models, one can learn the network by estimating the inverse covariance matrix using the graphical LASSO \citep{friedman2008sparse}; one can then identify changes in the inverse covariance matrix associated with disease condition \citep[][]{zhao2014direct, xia2015testing, he2019statistical}.
Alternatively, the condition-specific networks can be estimated using neighborhood selection \citep{meinshausen2006high}; in this approach, partial correlations among nodes are estimated by fitting a series of linear regressions in which one node is treated as the outcome, and the remaining nodes are treated as regressors. Changes in the network can then be delineated from differences in the regression coefficients by disease condition \cite[][]{belilovsky2016testing, xia2018two}.
More generally, the condition-specific networks are often modeled using exponential family pairwise interaction models \cite{lin2016estimation, yang2015graphical, yu2019generalized, yu2020simultaneous}.

The approach to differential network analysis described above may lead to the detection of between-group differences in biological networks that are not necessarily meaningful, in particular, when the condition-specific networks depend on covariates (e.g., age and sex).
This is because between-group network differences can be induced by \textit{confounding variables}, i.e., variables that are associated with both the within-group networks, and the disease condition.
In such cases, the network differences by disease condition may only reflect the association between the confounding variable and the disease.
It is therefore important to account for the relationship between covariates and biological networks when performing differential network analysis.

In this paper, we propose a two-sample test for differential network analysis that accounts for within-group dependence of the networks on covariates.
More specifically, we propose to perform covariate-adjusted inference using a class of pairwise interaction models for the within-group networks.
Our approach treats each condition-specific network as a function of the covariates.
It then performs a hypothesis test for equivalence of these functions.
To accommodate the high-dimensional setting, in which the number of nodes in the network is large relative to the number of samples collected, we propose to estimate the networks using a regularized estimator and to perform hypothesis testing using a bias-corrected version of the regularized estimate \citep{van2016estimation}.

Our proposal is related to existing literature on modeling networks as functions of a small number of variables.
For example, there are various proposals for estimating high-dimensional inverse covariance matrices, conditional upon continuous low-dimensional features \citep{zhou2010time,wang2014inference}.
Also related are methods for regularized estimation of high-dimensional varying coefficient models, wherein the regression coefficients are functions of a small number of covariates  \cite{wang2009shrinkage}.
Our method is similar but places a particular emphasis on hypothesis testing in order to assess the statistical significance of observed changes in the network.
Our approach lays the foundation for a general class of graphical models and is the first, to the best of our knowledge, to perform covariate-adjusted hypothesis tests for differential network analysis.

The rest of the paper is organized as follows.
In Section 2, we begin with a broad overview of our proposed framework for covariate-adjusted  differential network analysis in  pairwise interaction exponential family models and introduce some working examples.
In the following sections, we specialize our framework by considering two different approaches for estimation and inference:
In Section 3, we describe a method that uses neighborhood selection \citep{meinshausen2006high, chen2015selection, yang2015graphical}, and in Section 4, we discuss an alternative estimation approach that utilizes the score matching framework of Hyv\"arinen \citep{hyvarinen2005estimation, hyvarinen2007some}.
We assess the performance of our proposed methodology on synthetic data in Section 5 and apply it to a breast cancer data set from The Cancer Genome Atlas (TCGA) \citep{weinstein2013cancer} in  Section 6.
We conclude with a brief discussion in Section 7.

%

\section{Overview of the Proposed Framework}

\subsection{Differential Network Analysis without Covariate Adjustment}

To formalize our problem, we begin by introducing some notation.
We compare networks between two groups, labeled by $g \in \{\mathrm{I},\mathrm{II}\}$.
We obtain measurements of $p$ variables $X^g = \left(X^g_1,\ldots,X_p^g\right)^\top$, corresponding to nodes in a  graphical model \citep{maathuis2018handbook}, on $n^{\mathrm{I}}$ subjects in group $\mathrm{I}$ and $n^{\mathrm{II}}$ subjects in group $\mathrm{II}$.
We define $\mathcal{X} \subseteq \mathds{R}^p$ as the sample space of $X^g$.
Let $X^g_{i,j}$ denote the data for node $j$ for subject $i$ in group $g$, and let  $\mathbf{X}_j^g = (X^g_{1,j}, \ldots, X^g_{n^g,j})^\top$ be an $n^{g}$-dimensional vector of measurements on node $j$ for group $g$.

Our objective is to determine whether the association between variables $X_j$ and $X_k$, conditional upon all other variables, differs by group.
Our approach is to specify a model for $X^g$ such that the conditional dependence between any two nodes $X^g_{j}$ and $X^g_{k}$ can be represented by a single scalar parameter $\beta^{g,*}_{j,k}$.
If the association between nodes $j$ and $k$ is the same in both groups $\mathrm{I}$ and $\mathrm{II}$,  $\beta^{\mathrm{I},*}_{j,k} = \beta^{\mathrm{II},*}_{j,k}$.
Conversely, if $\beta^{\mathrm{I},*}_{j,k} \neq \beta^{\mathrm{II},*}_{j,k}$, we say nodes $j$ and $k$ are \textit{differentially connected}.
We assess for differential connectivity by performing a test of the null hypothesis
\begin{align}
H^0_{j,k}: \beta^{\mathrm{I},*}_{j,k} =\beta^{\mathrm{II},*}_{j,k}.
\label{NullUnadj}
\end{align}

We consider a general class of exponential family pairwise interaction models.
For $x = (x_1,\ldots,x_p)^\top$, we assume the density function for $X^g$ takes the form
\begin{align}
f^{g,*}(x) = \exp\left(\sum_{j=1}^p \mu_j(x_j) + \sum_{j = 1}^p \sum_{k = 1}^{j} \beta^{g,*}_{j,k}\psi_{j,k}(x_j, x_k) - U\left(\boldsymbol{\beta}^{g,*} \right)  \right),
\label{ExpFam}
\end{align}
where $\psi_{j,k}$ and $\mu_j$ are fixed and known functions,  $\boldsymbol{\beta}^{g,*}$ is a $p\times p$ matrix with elements $\beta^{g,*}_{j,k}$, and $U(\boldsymbol{\beta}^{g,*})$ is the log-partition function.
The dependence between $X^g_{j}$ and $X^g_{k}$ is measured by $\beta^{g,*}_{j,k}$, and nodes $j$ and $k$ are conditionally independent in group $g$ if and only if $\beta^{g,*}_{j,k} = 0$.

This class of exponential family distributions is rich and includes several models that have been studied previously in the graphical modeling literature.
One such example is the Gaussian graphical model, perhaps the most widely-used graphical model for continuous data.
For $x \in \mathds{R}^p$ the density function for mean-centered Gaussian random vectors can be expressed as
\begin{align}
f^{g,*}(x) \propto \exp\left( -\sum_{j = 1}^p \sum_{k = 1}^j \beta^{g,*}_{j,k} x_jx_k \right),
\label{ExpFamGGM}
\end{align}
and is thus a special case of \eqref{ExpFam} with $\psi_{j,k} = -x_jx_k$ and $\mu_j = 0$.
The non-negative Gaussian density, which takes the form of \eqref{ExpFamGGM} with the constraint that $x$ takes values in $\mathds{R}^p_+$, also belongs to the exponential family class.
Another canonical example is the Ising model, commonly used for studying conditional dependencies among binary random variables.
For $x \in \{0,1\}^p$, the density function for the Ising model can be expressed as
\begin{align*}
f(x) \propto \exp\left(\sum_{j = 1}^p \sum_{k=1}^j \beta_{j,k} x_jx_k\right).
\end{align*}
Additional examples include the Poisson model, the exponential graphical model, and conditionally-specified mixed graphical models \citep{yang2015graphical, chen2015selection}.

When asymptotically normal estimates of $\beta^{\mathrm{I},*}_{j,k}$ and $\beta^{\mathrm{II},*}_{j,k}$  are available, one can perform a calibrated test of $H^0_{j,k}$ based on the difference between the estimates.
In many cases, asymptotically normal estimates can be obtained using well-established methodology.
For instance, when the log-partition function $U(\boldsymbol{\beta}^{g,*})$ is available in closed form and is tractable, one can obtain estimates via (penalized) maximum likelihood.
This is a standard approach in the Gaussian setting, in which case the log-partition function is easy to compute.
However, this is not the case for other exponential family models.
Likelihood-based estimation strategies are thus generally difficult to implement.
In this paper, we consider two alternative strategies that have been proposed to overcome these computational challenges and are more broadly applicable.

The first approach we discuss is neighborhood selection \citep{chen2015selection, meinshausen2006high, yang2015graphical}.
Consider a sub-class of exponential family graphical models for which the conditional density function for any node $X^g_{j}$ given the remaining nodes belongs to a univariate exponential family model.
Because the log-partition function in univariate exponential family models is available in closed form, it is computationally feasible to estimate each conditional density function.
By estimating the conditional density functions, one can identify the \textit{neighbors} of nodes $j$, that is, the nodes upon which the conditional distribution depends.
This approach was first proposed as an alternative to maximum likelihood estimation for estimating Gaussian graphical models \citep{meinshausen2006high}.
To describe our approach, we focus on the Gaussian case, though this approach is more widely applicable and can be used for modeling dependencies among, e.g., Poisson, binomial, and exponential random variables as well \citep{chen2015selection, yang2015graphical}.


In Gaussian graphical models, the dependency of node $j$ on all other nodes can be determined based on the linear model 
\begin{align}
\mathds{E}\left[X^g_{j}| X^g_{1},\ldots,X^g_{p}\right] = \beta^{g,*}_{j,0} + \sum_{k \neq j}\beta^{g,*}_{j,k}X^g_{k}.
\label{UnadjustedModel}
\end{align}
The regression coefficients $\beta^{g,*}_{j,k}$ measure the strength of linear association between nodes $j$ and $k$ conditional upon all other nodes and are zero if and only if nodes $j$ and $k$ are conditionally independent; $\beta^{g,*}_{j,0}$ is an intercept term and is zero if all nodes are mean-centered.
(We acknowledge a slight abuse of notation here, as the regression coefficients in \eqref{UnadjustedModel} are not equivalent to parameters in \eqref{ExpFam}.
However, either estimand fully characterizes conditional independence.)
In the low-dimensional setting (i.e., $p \ll n^g$), statistically efficient and asymptotically normal estimates of the regression coefficients can be readily obtained  via ordinary least squares.
In high-dimensions (i.e., $p \geq n^g$), the ordinary least squares estimates are inconsistent, so to obtain consistent estimates we typically rely upon regularized estimators such as the LASSO and the elastic net \citep{tibshirani1996regression, zou2005regularization}.
Regularized estimators are generally biased and have intractable sampling distributions, and as such, are unsuitable for performing formal statistical inference.
However, several methods have recently emerged for obtaining asymptotically normal estimates by correcting the bias of regularized estimators \citep{javanmard2014confidence, van2014asymptotically, zhang2014confidence}.

The second computationally efficient approach we consider is to estimate the density function using the score matching framework of Hyv\"arinen \citep{hyvarinen2005estimation, hyvarinen2007some}.
Hyv\"arinen derives a loss function for estimation of density functions for continuous random variables that is based on the gradient of the log-density with respect to the observations.
As such, the score matching loss does not depend on the log-partition function in exponential family models.
Moreover, when the joint distribution for $X^g$ belongs to an exponential family model, the loss is quadratic in the unknown parameters, allowing for efficient computation.
In low dimensions, the minimizer of the score matching loss is consistent and asymptotically normal.
In high dimensions, one can obtain asymptotically normal estimates by minimizing a regularized version of the score matching loss to obtain an initial estimate \citep{lin2016estimation, yu2019generalized} and subsequently correcting for the bias induced by regularization \citep{yu2020simultaneous}.

\subsection{Covariate-Adjusted Differential Network Analysis}

We now consider the setting in which the within-group networks depend on covariates.
We denote by $W^g$ a $q$-dimensional random vector of covariate measurements for group $g$, and we define $\mathcal{W}$ as the sample space of $W^g$.
Let $W^g_{i,r}$ refer to the value of covariate $r$ for subject $i$ in group $g$, and let $W^g_i = (W^g_{i,1},\ldots,W^g_{i,q})^\top$ be a $q$-dimensional vector containing all covariates for subject $i$ in group $g$.
We assume the number of covariates is small relative to the sample size (i.e., $q \ll n^g$).

To study the dependence of the within-group networks on the covariates,  we specify a model for the nodes $X^g$ given the covariates $W^g$ that allows for the inter-node dependencies to vary as a function of $W^g$.
The model defines a function $\eta^g_{j,k}: \mathcal{W} \to \mathds{R}$ that takes as input a vector of covariates and returns a measure of association between nodes $j$ and $k$ for a subject in group $g$ with identical covariates.
One can interpret $\eta^{g,*}_{j,k}$ as a conditional version of $\beta^{g,*}_{j,k}$, given the covariates.

We assume that $\eta^{g,*}_{j,k}$ can be written as a low-dimensional linear basis expansion in $W^g$ of dimension $d$ --- that is,
\begin{align}
\eta^{g,*}_{j,k}\left(W^g\right) = \left\langle \phi\left(W^g\right), \alpha_{j,k}^{g,*} \right\rangle,
\label{BasisExp}
\end{align}
where $\phi:\mathds{R}^q\to\mathds{R}^d$ is a map from a set of covariates to its expansion,  $\alpha_{j,k}^{g,*}$ is a $d$-dimensional vector, and $\langle \cdot, \cdot \rangle$ denotes the vector inner product.
Let $\phi_c(w)$ refer to the $c$-th element of $\phi(w)$.
One can take the simple approach of specifying $\phi$ as a linear basis, $\phi(w) = \left(1, w_1, \ldots, w_q\right)$ for $w \in \mathds{R}^q$, though more flexible choices such as polynomial or B-spline bases can also be considered.
It may be preferable to specify $\phi$ so that $\eta^{g,*}_{j,k}$ is an additive function of the covariates.
This allows one to easily assess the effect of any specific covariate on the network by estimating the sub-vector of $\alpha^{g,*}_{j,k}$ that is relevant to the covariate of interest.

When the association between nodes $j$ and $k$ does not depend on group membership, $\eta^{\mathrm{I,*}}_{j,k}(w) = \eta^{\mathrm{II,*}}_{j,k}(w)$ for all $w$, and $\alpha^{\mathrm{I,*}}_{j,k} = \alpha^{\mathrm{II,*}}_{j,k}$.
In other words, if one subject from group $\mathrm{I}$ and another subject from group $\mathrm{II}$ have identically-valued covariates, the corresponding measure of association between nodes $j$ and $k$ is also the same.
In the covariate-adjusted setting, we say that nodes $j$ and $k$ are differentially connected if there exists $w$ such that $\eta^{\mathrm{I,*}}_{j,k}(w) \neq \eta^{\mathrm{II,*}}_{j,k}(w)$, or equivalently, if $\alpha^{\mathrm{I,*}}_{j,k} \neq \alpha^{\mathrm{II,*}}_{j,k}$.
We can thus assess differential connectivity between nodes $j$ and $k$ by testing the null hypothesis
\begin{align}
G^{0}_{j,k}: \alpha^{\mathrm{I,*}}_{j,k} = \alpha^{\mathrm{II,*}}_{j,k}.
\label{NullAdj}
\end{align}
Similar to the unadjusted setting, when asymptotically normal estimates of $\alpha^{\mathrm{I},*}_{j,k}$ and $\alpha^{\mathrm{II},*}_{j,k}$ are available, a calibrated test can be constructed based on the difference  between the estimates.

We now specify a form for the conditional distribution of $X^g$ given $W^g$ as a generalization of the exponential family pairwise interaction model \eqref{ExpFam}.
We assume the conditional density for $X^g$ given $W^g$ can be expressed as
\begin{align}
&f^{g,*}(x|w) \propto
\exp\left(\sum_{j=1}^p \mu_j(x_j) + \sum_{j = 1}^p \sum_{k=1}^j \eta^{g,*}_{j,k}(w)\psi_{j,k}(x_j, x_k) + \sum_{j=1}^p \sum_{c=1}^d \theta_{j,c}^{g,*} \zeta_{j,c}\left(x_j, \phi_c(w)\right) \right),
\label{ExpFam2}
\end{align}
where $w = (w_1,\ldots,w_q)^\top$, and the proportionality is up to a normalizing constant that does not depend on $x$.
Above, $\zeta_{j,c}$ is a fixed and known function, and the main effects of the covariates on $X^g$ are represented by the scalar parameters $\theta^{g,*}_{j,c}$. 
The conditional dependence between nodes $j$ and $k$, given all other nodes and given that $W^g = w$ is quantified by $\eta^{g,*}_{j,k}(w)$, and $\eta^{g,*}_{j,k}(w) = 0$ if and only if nodes $j$ and $k$ are conditionally independent at $w$. 
One can thus view $\eta^{g,*}_{j,k}$ as a conditional version of $\beta^{g,*}_{j,k}$ in \eqref{ExpFam2}.

%
Either of the estimation strategies introduced in Section 2.1 can be used to perform covariate-adjusted inference.
When the conditional distribution of each node given the remaining nodes \textit{and the covariates} belongs to a univariate exponential family model, the covariate-dependent network can be estimated using neighborhood selection because the node conditional distributions can be estimated efficiently with likelihood-based methods.
Alternatively, we can estimate the conditional density function \eqref{ExpFam2} using score matching.

As a working example, we again consider estimation of covariate-dependent Gaussian networks using neighborhood selection.
Suppose the conditional distribution of $X^g$ given $W^g$ takes the form
\begin{align}
&f^{g,*}(x|w) \propto
\exp\left(-\sum_{j = 1}^p \sum_{k=1}^j  \eta^{g,*}_{j,k}(w)x_jx_k - \sum_{j=1}^p \sum_{c=1}^d \theta_{j,c}^{g,*} x_j \phi_c(w) \right).
\label{CondGauss}
\end{align}
Then the dependencies of node $j$ on all other nodes can be determined based on the following varying coefficient model \citep{hastie1993varying}:
\begin{align}
\mathds{E}\left[X^g_{j}|X_{1}^g,\ldots,X_{p}^g, W^g\right] = \eta_{j,0}^{g,*}(W^g)  + \sum_{k \neq j} \eta^{g,*}_{j,k}\left(W^g\right) X_{k}^g.
\label{AdjModel1}
\end{align}
The varying coefficient model is a generalization of the linear model that treats the regression coefficients as functions of the covariates.
In \eqref{AdjModel1}, $\eta^{g,*}_{j,k}(w)$ returns a regression coefficient that quantifies the linear relationship between nodes $j$ and $k$ for subjects in group $g$ with covariates equal to $w$.
Then $X^g_{j}$ and $X^g_{k}$ are conditionally independent given all other nodes and given $W^g = w$ if and only if $\eta^{g,*}_{j,k}(w) = 0$. 
The varying coefficients $\eta^{g,*}_{j,k}$ can thus be viewed as a conditional version of the regression coefficients in \eqref{UnadjustedModel}.
(We have again abused the notation, as the varying coefficient functions in \eqref{AdjModel1} are not equal to the parameters in \eqref{CondGauss}, though both functions are zero for the same values of $w$).
The intercept term $\eta^{g,*}_{j,0}$ accounts for the main effect of $W^{g}$ on $X^g_{j}$.
We can remove this main effect term by first centering the nodes $X^g_{j}$ about their conditional mean given $W^g$ (which can be estimated by performing a linear regression of $X^g_{j}$ on $\phi(W^g)$).

In Sections 3 and 4, we discuss construction of asymptotically normal estimators of $\alpha^{g,*}_{j,k}$ in the low- and high-dimensional settings using neighborhood selection and score matching.
Before proceeding, we first examine the connection between the null hypotheses $H^0_{j,k}$ and $G^0_{j,k}$.

\subsection{The Relationship between Hypotheses $H^0_{j,k}$ and $G^0_{j,k}$}

Hypotheses $H^0_{j,k}$ in \eqref{NullUnadj} and $G^0_{j,k}$ in \eqref{NullAdj} are related but not equivalent.
It is possible that $H^0_{j,k}$ holds while $G^0_{j,k}$ fails and vice versa.
We provide an example below.
Suppose we are using neighborhood selection to perform differential network analysis in the Gaussian setting, so we are making a comparison of linear regression coefficients between the two groups.
Suppose further that the within-group networks depend on single scalar covariate $W^g$, and the nodes are centered about their conditional mean given $W^g$.
One can show that the regression coefficients $\beta^{g,*}_{j,k}$ are equal to the average of their conditional versions $\eta^{g,*}_{j,k}(W^g)$.
That is, $\beta^{g,*}_{j,k} = \mathds{E}[\eta^{g,*}_{j,k}(W^g)]$.
Now, suppose $G^0_{j,k}$ holds.  
If $W^{\mathrm{I}}$ and $W^{\mathrm{II}}$ do not share the same distribution (e.g., the covariate tends to take higher values in group $\mathrm{I}$ than in group $\mathrm{II}$), the average conditional inter-node association may differ, and $H^0_{j,k}$ may not hold.
Although the conditional association between nodes, given the covariate, does not differ by group, the \textit{average} conditional association does differ, as illustrated in Figure \ref{fig:ConfAndPower}a.
In such a scenario, the difference in the average conditional association is induced by the dependence of the covariate on group membership and the dependence of the inter-node association on the covariate.
Thus, inequality of $\beta^{\mathrm{I},*}_{j,k}$ and $\beta^{\mathrm{II},*}_{j,k}$ does not necessarily capture a meaningful association between the network and group membership.
Similarly when $H^0_{j,k}$ holds, it is possible that $\eta^{\mathrm{I},*}_{j,k} \neq \eta^{\mathrm{II},*}_{j,k}$.
For instance, suppose that the distribution of the covariate is the same in both groups, and $\mathds{E}[\eta^g(W^g)] = 0$ in both groups.
If the between-node association depends more strongly upon the covariates in one group than the other,  $G^0_{j,k}$ will be false.
This example is depicted in Figure \ref{fig:ConfAndPower}b.
In this scenario, adjusting for covariates should provide improved power to detect differential connections.
We note that for other distributions, it does not necessarily hold that $\beta^{g,*}_{j,k} = \mathds{E}[\eta^{g,*}_{j,k}(W^g)]$, but regardless, there is generally no equivalence between hypotheses $H^0_{j,k}$ and $G^0_{j,k}$.

\begin{figure}[!h]
\centering
\begin{subfigure}{1\textwidth}
  \centering
  \includegraphics[width=.6\linewidth]{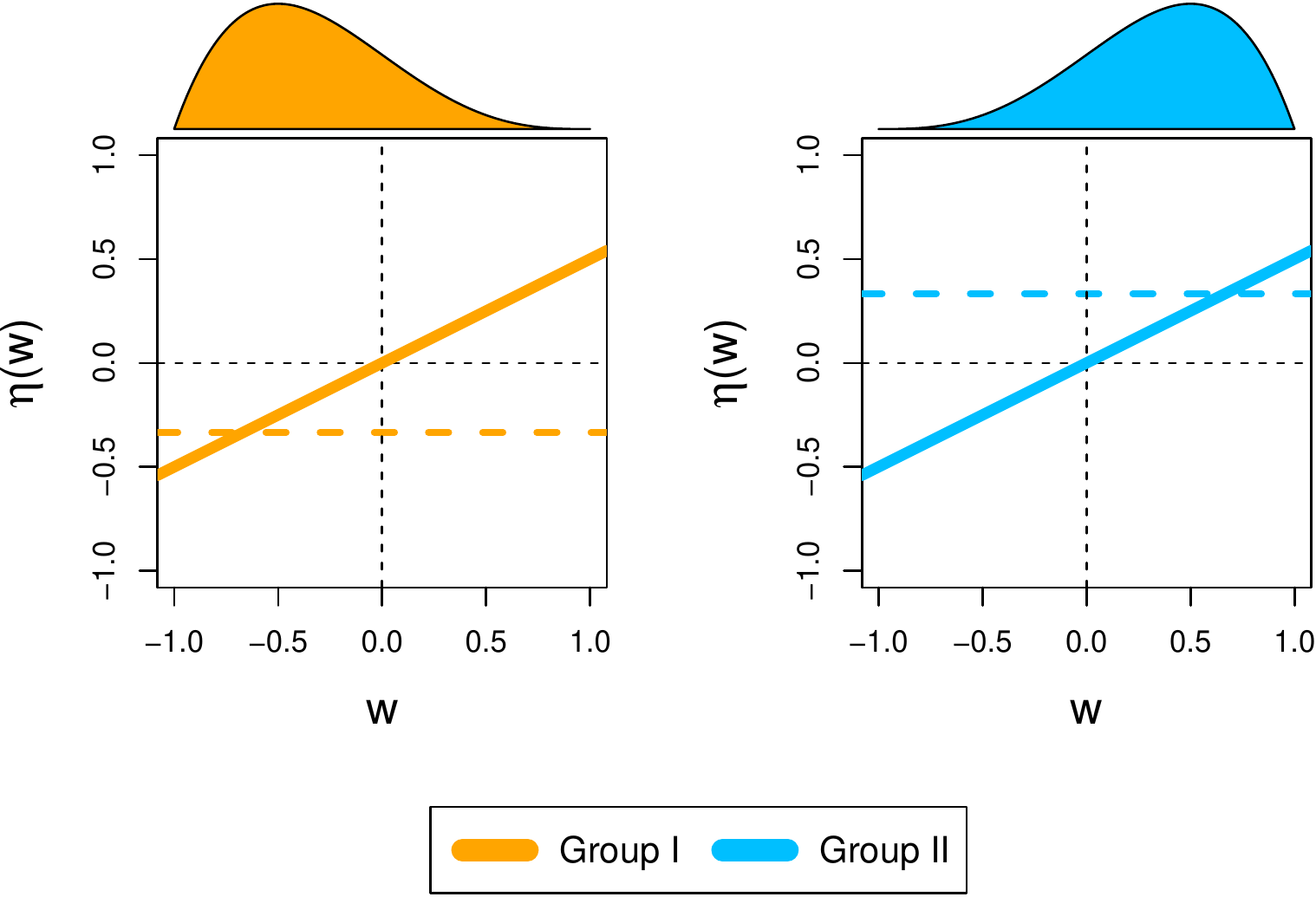}
\caption{}
  \label{fig:sub1}
\end{subfigure}%
\\
\begin{subfigure}{1\textwidth}
  \centering
  \includegraphics[width=.6\linewidth]{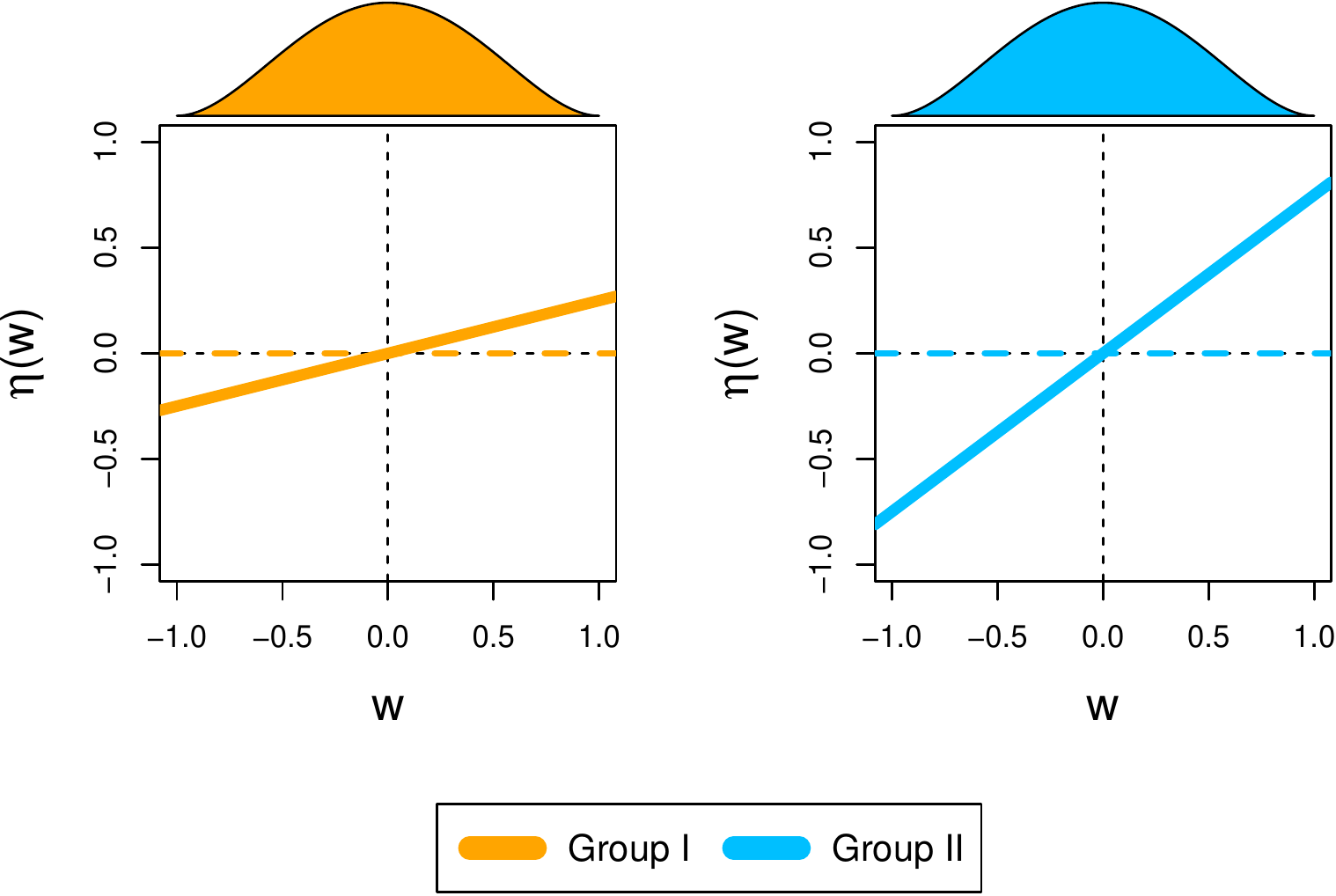}
\caption{}
  \label{fig:sub2}
\end{subfigure}
\caption{Displayed are the association between nodes $j$ and $k$, $\eta^g_{j,k}(\cdot)$, as a function of covariate $W^g$ and the distribution of $W^g$ in groups $\mathrm{I}$ and $\mathrm{II}$.  The average inter-node association is represented by the dashed colored lines. In (a), the average inter-node association depends on group membership, though the inter-node association given the covariate does not. In (b), the average inter-node association does not depend on group membership, though the conditional association between nodes given the covariate does depend on group membership.}
\label{fig:ConfAndPower}
\end{figure}


\section{Covariate-Adjusted Differential Network Analysis Using Neighborhood Selection}

In this section, we describe in detail an approach for covariate-adjusted differential network analysis using neighborhood selection.
To simplify our presentation, we focus on Gaussian graphical models, though this strategy is generally applicable to graphical models for which the node conditional distributions belong to univariate exponential family models.

\subsection{Covariate-Adjustment via Neighborhood Selection in Low Dimensions}

We first discuss testing the unadjusted null hypothesis $H^0_{j,k}$ in \eqref{NullUnadj}, where the $\beta^{g,*}_{j,k}$ are the regression coefficients in \eqref{UnadjustedModel}.
Suppose, for now, that we are in the low-dimensional setting, so the number of nodes $p$ is smaller than the sample sizes $n^g$, $g \in \{\mathrm{I}, \mathrm{II}\}$.

It is well-known that the regression coefficients can be characterized as the minimizers of the  expected least squares loss --- that is,
\begin{align*}
\boldsymbol{\beta}_j^{g,*} = (\beta^{g,*}_{j,1},\ldots,\beta^{g,*}_{j,p})^\top =  \underset{\beta_1,\ldots,\beta_p \in \mathds{R}}{\text{arg min}} \, \mathds{E} \left[ \left(X^g_{j} - \sum_{k \neq j} X^g_{k}\beta_k\right)^2\right].
\end{align*}
One can obtain an estimate $\hat{\boldsymbol{\beta}}_j^g = (\hat{\beta}^g_{j,1},\ldots,\hat{\beta}^g_{j,p})$ of $\boldsymbol{\beta}^{g,*}_j = (\beta^{g,*}_{j,1}\ldots.,\beta^{g,*}_{j,p})$ by minimizing the empirical average of the least squares, taking
\begin{align*}
\hat{\boldsymbol{\beta}}^g_j = \underset{\beta_1,\ldots,\beta_p \in \mathds{R}}{\text{arg min}} \,\frac{1}{n^g}\left\|\mathbf{X}^g_j - \sum_{k \neq j} \mathbf{X}^g_k\beta_k  \right\|_2^2, 
\end{align*}
where $\|\cdot\|_2$ denotes the $\ell_2$ norm.
The ordinary least squares estimate $\hat{\boldsymbol{\beta}}^g_{j}$ is available in closed form and is easy to compute.
The  estimates $\hat{\beta}^g_{j,k}$ are unbiased, and, under mild assumptions, are approximately normally distributed for sufficiently large $n^g$ --- that is,
\begin{align*}
\hat{\beta}^{g}_{j,k} \sim N\left(\beta^{g,*}_{j,k}, \tau^g_{j,k} \right),
\end{align*}
with $\tau^g_{j,k} > 0$ (though $\tau^g_{j,k}$ can be calculated in closed form, we omit the expression for brevity).

We construct a test of $H_{j,k}^0$ based on the difference between the estimates of the group-specific regression coefficients, $\hat{\beta}^{\mathrm{I}}_{j,k} - \hat{\beta}^{\mathrm{II}}_{j,k}$.
When $H^0_{j,k}$ holds, $\hat{\beta}^{\mathrm{I}}_{j,k} - \hat{\beta}^{\mathrm{II}}_{j,k}$ is normally distributed with mean zero and variance $\tau^{\mathrm{I}}_{j,k} + \tau^{\mathrm{II}}_{j,k}$.
Given a consistent estimate $\hat{\tau}^g_{j,k}$ of the variance, we can use the test statistic
\begin{align*}
T_{j,k} = \frac{\left(\hat{\beta}_{j,k}^{\mathrm{I}} - \hat{\beta}^{\mathrm{II}}_{j,k}\right)^2}{\hat{\tau}^{\mathrm{I}}_{j,k} + \hat{\tau}^{\mathrm{II}}_{j,k}},
\end{align*}
which follows a chi-square distribution with one degree of freedom under the null for $n^{\mathrm{I}}$ and $n^{\mathrm{II}}$ sufficiently large.
A p-value for $H^0_{j,k}$ can be calculated as
\begin{align*}
\rho_{j,k} = P\left( \chi^2_1 > T_{j,k} \right).
\end{align*}



In the low-dimensional setting, performing a covariate-adjusted test is similar to performing the unadjusted test.
We can obtain an estimate $\hat{\boldsymbol{\alpha}}^g_j = \left((\hat{\alpha}^g_{j,1})^\top,\ldots,(\hat{\alpha}^g_{j,p})^\top\right)^\top$ of $\boldsymbol{\alpha}^{g,*}_j = \left((\alpha^{g,*}_{j,1})^\top,\ldots,(\alpha^{g,*}_{j,p})^\top\right)^\top$ by minimizing the empirical average of the least squares loss
\begin{align}
\hat{\boldsymbol{\alpha}}^g_j = \underset{\alpha_{j,1},\ldots,\alpha_{j,p} \in \mathds{R}^{d}}{\text{arg min}} \,\frac{1}{n^g}\sum_{i=1}^{n^g} \left(X^g_{i,j} - \sum_{k \neq j} \left\langle \phi\left(W_i^g\right), \alpha_{j,k} \right\rangle X_{i,k}^g\right)^2.
\label{AlphaOLS1}
\end{align}
To simplify the presentation, we introduce additional notation that allows us to rewrite \eqref{AlphaOLS1} in a condensed form. 
Let $\mathcal{V}^g_k$ be the $n^g \times d$ matrix
\begin{align}
\mathcal{V}^g_k =
\begin{pmatrix}
X_{1,k}^g \times \phi\left(W_1^g\right) 
\\
\vdots 
\\
X_{n^g, k}^g \times \phi\left(W^g_{n^g}\right)
\end{pmatrix}.
\label{VjLinear}
\end{align}
We can now equivalently express \eqref{AlphaOLS1} as
\begin{align}
\hat{\boldsymbol{\alpha}}^g_j =  \underset{\alpha_{j,1},\ldots,\alpha_{j,p} \in \mathds{R}^{d}}{\text{arg min}} \,\frac{1}{n^g}\left\|\mathbf{X}^g_j - \sum_{k \neq j}  \mathcal{V}^g_k \alpha_{j,k}\right\|_2^2.
\label{AlphaOLS2}
\end{align}
Again, $\hat{\alpha}^g_{j,k}$ is an unbiased and approximately normal for sufficiently large $n^g$, satisfying
\[
\hat{\alpha}^g_{j,k} \sim N\left(\alpha^{g,*}_{j,k}, \Omega^g_{j,k}\right),
\]
where $\Omega^g_{j,k}$ is a positive definite matrix of dimension $d\times d$ (though a closed form expression is available, we omit it here for brevity).

We construct a test of $G^0_{j,k}$ based on $\hat{\alpha}^{\mathrm{I}}_{j,k} - \hat{\alpha}^{\mathrm{II}}_{j,k}$.
Under the null hypothesis, $\hat{\alpha}^{\mathrm{I}}_{j,k} - \hat{\alpha}^{\mathrm{II}}_{j,k}$ follows a normal distribution with mean zero and variance $\Omega^{\mathrm{I}}_{j,k} + \Omega^{\mathrm{II}}_{j,k}$.
Given a consistent estimate $\hat{\Omega}^g_{j,k}$ of $\Omega^g_{j,k}$, we can test $G^0_{j,k}$ using the test statistic
\begin{align*}
S_{j,k} = 
\left(\hat{\alpha}_{j,k}^{\mathrm{I}} - \hat{\alpha}_{j,k}^{\mathrm{II}}\right)^\top
\left(\hat{\Omega}^{\mathrm{I}}_{j,k} + \hat{\Omega}^{\mathrm{II}}_{j,k} \right)^{-1} 
\left(\hat{\alpha}_{j,k}^{\mathrm{I}} - \hat{\alpha}_{j,k}^{\mathrm{II}}\right).
\end{align*}
Under the null, the test statistic follows a chi-squared distribution with $d$ degrees of freedom, and a p-value can therefore be calculated as
\begin{align*}
P\left(\chi^2_{d} > S_{j,k}\right).
\end{align*}

\subsection{Covariate-Adjustment via Neighborhood Selection in High Dimensions}

The methods described in Section 3.1 are only appropriate when the number of nodes $p$ is small relative to the sample size.
Model \eqref{AdjModel1} has $(p-1)d$ parameters, so the least squares estimator of Section 3.1 provides stable estimates as long as $n^{\mathrm{I}}$ and $n^{\mathrm{II}}$ are larger than $(p-1)d$.
However, in the high-dimensional setting, where the the number of parameters exceeds the sample size, the ordinary least squares estimates behave poorly.

To fit the varying coefficient model \eqref{AdjModel1} in the high-dimensional setting, we use a regularized estimator that relies upon an assumption of \textit{sparsity} in the networks.
The sparsity assumption requires that within each group only a small number of nodes are partially correlated, meaning that in \eqref{AdjModel1}, only a few of the vectors $\alpha^{g,*}_{j,k}$ are nonzero.
To leverage the sparsity assumption, we propose to use the group LASSO estimator \citep{yuan2006model}:
\begin{align}
\tilde{\boldsymbol{\alpha}}_j^g = \underset{\alpha_{j,1},\ldots,\alpha_{j,p} \in \mathds{R}^{d}}{\text{arg min}} \,
\frac{1}{n^g}\left\|\mathbf{X}^g_j -  \sum_{k \neq j} \mathcal{V}^g_k \alpha_{j,k} \right\|_2^2 + \lambda \sum_{k \neq j} \left\| \alpha_{j,k} \right\|_2,
\label{GrpLASSO}
\end{align}
where $\lambda > 0$ is a tuning parameter.
The group LASSO provides a sparse estimate and sets some $\tilde{\alpha}_{j,k}$ to be exactly zero, resulting in networks with few edges.
The level of sparsity of $\tilde{\boldsymbol{\alpha}}^g_{j}$ is determined by $\lambda$, with higher $\lambda$ values forcing more $\tilde{\alpha}_{j,k}$ to zero.
We discuss selection of the tuning parameter in Section 5.1.  

Though the group LASSO provides a consistent estimate of $\boldsymbol{\alpha}^{g,*}_j$, the estimate is \textit{not} approximately normally distributed.
The group LASSO estimate of ${\alpha}^{g,*}_{j,k}$ retains a bias that diminishes at the same rate as the standard error.
As a result, the group LASSO estimator has a non-standard sampling distribution that cannot be derived analytically and is therefore unsuitable for hypothesis testing.

We can obtain approximately normal estimates of $\alpha^{g,*}_{j,k}$ by correcting the bias of $\tilde{\alpha}^g_{j,k}$, as was first proposed to obtain normal estimates for the classical $\ell_1$-penalized version of the LASSO \citep{van2014asymptotically, zhang2014confidence}.
These ``de-biased'' or ``de-sparsified'' estimators can been shown to be approximately normal with moderately large samples even in the high-dimensional setting; they are therefore suitable for hypothesis testing.
Our approach is to use a de-biased version of the group LASSO.
Bias correction in group LASSO problems is well studied \citep{van2016estimation, honda2019biased, mitra2016benefit}, so we are able to perform covariate-adjusted inference by applying previously-developed methods.

The bias of the group LASSO estimate can be written as
\begin{align}
\delta^g_{j,k} = \mathds{E}\left[\tilde{\alpha}^g_{j,k}\right] - \alpha^{g,*}_{j,k},
\end{align} 
where $\delta^g_{j,k}$ is a nonzero $d$-dimensional vector (recall $d$ is the dimension of $\alpha^{g,*}_{j,k}$).  
Our approach is to obtain an estimate of the bias $\tilde{\delta}_{j,k}$ and to use a de-biased estimator, defined as
\begin{align}
\check{\alpha}^g_{j,k} = \tilde{\alpha}^g_{j,k} - \tilde{\delta}^g_{j,k}.
\end{align}
For a suitable choice of $\tilde{\delta}_{j,k}$, the bias-corrected estimator is approximately normal for a sufficiently large sample size $n^g$ under mild conditions, i.e.,
\begin{align}
\check{\alpha}^g_{j,k} \sim N\left(\alpha^{g,*}_{j,k}, \Omega^g_{j,k}\right),
\end{align}
where the variance $\Omega^g_{j,k}$ is a positive definite matrix, for which we obtain an estimate $\check{\Omega}^g_{j,k}$.
We provide a derivation for the bias-correction and the form of our variance estimate in Appendix A.

Similar to Section 2.1, we test the null hypothesis $G^0_{j,k}$ in \eqref{NullAdj} using the test statistic
\begin{align}
S_{j,k} = 
\left(\check{\alpha}_{j,k}^{\mathrm{I}} - \check{\alpha}_{j,k}^{\mathrm{II}}\right)^\top
\left(\check{\Omega}_{j,k}^{\mathrm{I}} + \check{\Omega}_{j,k}^{\mathrm{II}} \right)^{-1} 
\left(\check{\alpha}_{j,k}^{\mathrm{I}} - \check{\alpha}_{j,k}^{\mathrm{II}}\right).
\label{TestStat2}
\end{align}
The test statistic asymptotically follows a chi-squared distribution with $d$ degrees of freedom under the null hypothesis.


\section{Covariate-Adjusted Differential Network Analysis Using Score Matching}

In this section, we discuss covariate-adjustment using the score matching framework introduced in Section~2.
We first describe the score matching estimator in greater detail and then specialize the framework to estimation of pairwise exponential family graphical models in the low- and high dimensional settings. As shown later in this section, for exponential family distributions with continuous support, the score matching loss function is a quadratic function of parameters, providing a computationally-efficient framework for estimating graphical models.

%

\subsection{The Score Matching Framework}

We begin by providing a brief summary of the score matching framework \citep{hyvarinen2005estimation, hyvarinen2007some}. 
Let $Z \in \mathcal{Z} \subseteq \mathds{R}^p$ be a random vector generated from a distribution with density function $h^*$.
For any candidate density $h$, we denote the gradient and Laplician of the log-density by
\begin{align*}
\nabla \log {h}(z) =\left\{ \frac{\partial}{\partial z_j}\log h(x) \right\} \in \mathds{R}^p; \quad \quad \Delta \log h(z) = \sum_{j=1}^p \frac{\partial^2}{\partial z_j^2} \log h(z_j).
\end{align*}
The \textit{score matching loss} $L$ is defined as a measure of divergence between a candidate density function $h$ and the true density $h^{*}$:
\begin{align}
L(h) = \int \left\|\nabla \log h(z)  - \nabla \log h^{*}(z)  \right\|_2^2  h^{*}(z) dz = \mathds{E}\left[ \left\|\nabla \log h(Z)  - \nabla \log h^{*}(Z)  \right\|_2^2  \right].
\label{ScoreLoss1}
\end{align}
It is apparent that the score matching loss is minimized when $h = h^*$.
A natural approach to constructing an estimator for $h^*$ would then be to minimize the empirical score matching loss given observations $Z_1,\ldots,Z_n$, defined as
\begin{align*}
L_n(h) = \frac{1}{n} \sum_{i=1}^n \left \| \nabla \log h\left(Z_i\right) - \nabla \log h^*\left(Z_i\right) \right \|_2^2. 
\end{align*}
Because the score matching loss function takes as input the gradient of the log density function, the loss does not depend on the normalizing constant.
This makes score matching appealing when the normalizing constant is intractable.

The empirical loss seemingly depends on prior knowledge of $h^*$.
However, if $h(z)$ and $\|h(z)\|_2$ both tend to zero as $z$ approaches the boundary of $\mathcal{Z}$, a partial integration argument can be used to show that the score matching loss can be expressed as
\begin{align}
L(h) = \int \left\{ \Delta \log h(z) +  \frac{1}{2}\left \| \nabla \log h(z)  \right \|_2^2 \right\} h^*(z)dz + \text{const.},
\label{ScoreLoss2}
\end{align}
where `const.' is a term that does not depend on $h$.
We can therefore estimate $h^*$ by minimizing an empirical version of the score matching loss that does not depend on $h^*$.
We can express the empirical loss as
\begin{align*}
L_n(h) = \frac{1}{n} \sum_{i=1}^n \Delta \log h(Z_i) +  \frac{1}{2}\left \| \nabla \log h(Z_i)  \right \|_2^2.
\end{align*}

The score matching loss is particularly appealing for exponential family distributions with continuous support, as it leads to a quadratic optimization function \citep{lin2016estimation}. 
However, when $Z$ is non-negative, the arguments used to express \eqref{ScoreLoss1}  as \eqref{ScoreLoss2} fail because $h(z)$ and $\|\nabla h(z) \|_2$ do not approach zero at the boundary. 
We can overcome this problem by instead considering the \textit{generalized score matching framework} \citep{yu2019generalized, hyvarinen2007some} as an extension that is suitable for non-negative data.
Let $v_1,\ldots,v_p: \mathds{R}^+ \to \mathds{R}^+$ be positive and differentiable functions, let $v(z) = \left(v_1(z_1),\ldots,v_p(z_p)\right)^\top$, let $\dot{v}_j$ denote the derivative of $v_j$, and let $\circ$ denote the element-wise product operator.
The generalized score matching loss is defined as
\begin{align}
L(h) = \int \left\|\left\{\nabla \log h(z)  - \nabla \log h^{*}(z) \right\} \circ v^{1/2}(z) \right\|_2^2  h^{*}(z) dz,  
\label{GenScoreLoss1}
\end{align}
and is also minimized when $h = h^*$.
As for the original score matching loss \eqref{ScoreLoss1}, the generalized score matching loss seemingly depends on prior knowledge of $h^{*}$.
However, under mild technical conditions on $h$ and $v$ (see Appendix B.1), the loss in \eqref{GenScoreLoss1} can be rewritten as
\begin{align}
L(h) = 
\int\bigg[
\sum_{j=1}^p &\dot{v}_j(z_j)\left\{\frac{\partial \log h(z_j)}{\partial z_j} \right\} +
v_j(z_j)\left\{\frac{\partial^2 \log h(z)}{\partial z^2_j} \right\}  + 
\frac{1}{2}v_j(z_j)\left\{\frac{\partial \log h(z)}{\partial z_j} \right\}^2 
 \bigg]h^*(z)dz.
\label{GenScoreLoss2}
\end{align}
The generalized score matching loss thus no longer depends on $h^*$, and an estimator can be constructed by minimizing the empirical version of \eqref{GenScoreLoss2} with respect to $h$.
To this end, the original generalized score matching estimator considered $v_j(z_j)  = z_j^2$ \citep{hyvarinen2007some}.
In this case, it becomes necessary to estimate high moments of $h^*$, leading to poor performance of the estimator.
It has been shown that by instead taking $v$ as a slowly increasing function, such as $v_j(z_j) = \log(1 + v_j)$, one obtains improved theoretical results and better empirical performance \cite{yu2019generalized}.

\subsection{Covariate-Adjustment in High-Dimensional Exponential Family Models via Score Matching}

In this sub-section, we discuss construction of asymptotically normal estimators for the parameters of the exponential family pairwise interaction model \eqref{ExpFam2} using the generalized score matching framework.
To simplify our presentation, we consider the setting in which we are only interested in studying the connectedness between one node $X^g_{j}$ and all other neighboring nodes in the network.
To this end, it suffices to estimate the conditional density of $X^g_{j}$ given all other nodes and the covariates $W^g$.
A similar approach to the one we describe below can also be used to estimate the entire joint density \eqref{ExpFam2}. 
For simplicity, we assume that in \eqref{ExpFam2}, there exist functions $\psi$ and $\zeta$ such that $\psi = \psi_{j,k}$ for all $(j,k)$ and $\zeta_{j,c} = \zeta$ for all $(j,c)$, and that $\mu_j = 0$.
For $x = (x_1,\ldots,x_p)^\top$ and $w = (w_1,\ldots,w_q)^\top$ the conditional density can thus be expressed as
\begin{align}
f_j^{g,*}(x_j|x_1,\ldots,x_p,w) \propto
\exp\left( \sum_{ j = 1}^p \left\langle \alpha^{g,*}_{j,k}, \phi\left(w\right) \right\rangle\psi(x_j, x_k) + \sum_{c=1}^d \theta_{j,c}^{g,*} \zeta\left(x_j, \phi_c(w)\right)\right),
\label{CondExpFam}
\end{align}
where the density is up to a normalizing constant that does not depend on $x_j$.

We first explicitly define the score matching loss for the conditional density function \eqref{CondExpFam}.
Let $\boldsymbol{\alpha}^{g,*}_j = \left((\alpha^{g,*}_{j,1})^\top, \ldots ,(\alpha^{g,*}_{j,p})^\top\right)^\top$, and similarly let $\boldsymbol{\theta}^{g,*}_j = (\theta^{g,*}_{j,1}, \ldots ,\theta^{g,*}_{j,p})^\top$.
Let $\dot{\psi}$ and $\ddot{\psi}$ denote the first and second derivatives of $\psi$ with respect to $x_j$, and similarly, let $\dot{\zeta}$ and $\ddot{\zeta}$ denote the first and second derivatives of $\zeta$ with respect to $x_j$. We define a non-negative function $v_j: \mathds{R}_+ \to \mathds{R}_+$, and let $\dot{v}_j$ denote the first derivative of $v_j$.
Then for candidate parameters $\boldsymbol{\alpha}_j = \left(\alpha^\top_{j,1},\ldots,\alpha^\top_{j,p}\right)^\top$ and $\boldsymbol{\theta}_j = (\theta_{j,1},\ldots,\theta_{j,d})^\top$, the empirical generalized score matching loss for the conditional density of $X_{j}^g$ given all other nodes and the covariates can be expressed as
\begin{align}
L^g_{n,j}(\boldsymbol{\alpha}, \boldsymbol{\theta}) = \frac{1}{2n^g}&\sum_{i=1}^{n^g}  v_j\left(X_{i,j}^g\right)\bigg\{ \sum_{k  = 1}^p \left\langle \alpha_{j,k}, \phi\left(W^g_{i}\right) \right\rangle\dot{\psi}\left(X^g_{i,j}, X^g_{i,k}\right) + \sum_{c=1}^d \theta_{j,c} \dot{\zeta}\left(X^g_{i,j}, \phi_c\left(W^g_{i}\right)\right) \bigg\}^2 + \nonumber
\\
\frac{1}{n^g}&\sum_{i=1}^{n^g} v_j\left(X^g_{i,j}\right)\bigg\{ \sum_{k = 1}^p \left\langle \alpha_{j,k}, \phi\left(W^g_{i}\right) \right\rangle\ddot{\psi}\left(X^g_{i,j}, X^g_{i,k}\right) + \sum_{c=1}^d \theta_{j,c} \ddot{\zeta}\left(X^g_{i,j}, \phi_c\left(W^g_{i}\right)\right) \bigg\} + \nonumber
\\
\frac{1}{n^g}&\sum_{i=1}^{n^g} \dot{v}_j\left(X^g_{i,j}\right)\bigg\{ \sum_{k = 1}^p \left\langle \alpha_{j,k}, \phi\left(W^g_{i}\right) \right\rangle\dot{\psi}\left(X^g_{i,j}, X^g_{i,k}\right) + \sum_{c=1}^d \theta_{j,c} \dot{\zeta}\left(X^g_{i,j}, \phi_c\left(W^g_{i}\right)\right) \bigg\}.
\label{ExpFamScoreLoss}
\end{align}
The true parameters $\boldsymbol{\alpha}^{g,*}_j$ and $\boldsymbol{\theta}^{g,*}_{j}$ can characterized as the minimizers of the population score matching loss $\mathds{E}\left[L^g_{n,j}\left(\boldsymbol{\alpha}_j, \boldsymbol{\theta}_j\right)\right]$, as discussed in Section 4.1.

The loss function in \eqref{ExpFamScoreLoss} is quadratic in parameters  $\boldsymbol{\alpha}^{g,*}_j$ and $\boldsymbol{\theta}_j^{g,*}$ and can thus be solved efficiently. 
When the sample size $n^g$ is much larger than the number of unknown parameters $(p+1)d$, one can estimate $\boldsymbol{\alpha}^{g,*}_j$ and $\boldsymbol{\theta}_j^{g,*}$ by simply minimizing $L^g_{n,j}$ with respect to the unknown parameters.
The empirical loss function is quadratic in $(\boldsymbol{\alpha}_j, \boldsymbol{\theta}_j)$, so the minimizer of the loss is available in closed form and can be computed efficiently.
Moreover, we can readily establish asymptotic normality of the parameter estimates using results from classical M-estimation theory \citep{van2000asymptotic}.
To avoid including cumbersome notation, we reserve the details for Appendix B.2.

When the sample size is smaller than the number of parameters, the minimizer of $L^g_{n,j}$ is no longer consistent.
Similar to Section 3.2, we use regularization to obtain a consistent estimator in the high-dimensional setting.
We define the $\ell_2$-regularized generalized score matching estimator as
\begin{align}
\left(\tilde{\boldsymbol{\alpha}}^g_j, \tilde{\boldsymbol{\theta}}^g_j \right) = \underset{\boldsymbol{\alpha}_j, \boldsymbol{\theta}_j}{\text{arg min}}\,\,L^g_{n,j}(\boldsymbol{\alpha}_j, \boldsymbol{\theta}_j) + \lambda \sum_{j=1}^p \left\| \alpha_{j,k} \right\|_2,
\label{RegularizedScoreMatching}
\end{align}
where $\lambda > 0$ is a tuning parameter.
Similar to the group LASSO estimator \eqref{GrpLASSO}, the regularization term in \eqref{RegularizedScoreMatching} induces sparsity in the estimate $\tilde{\boldsymbol{\alpha}}_j^g$ and sets some $\tilde{\alpha}^g_{j,k}$ to be exactly zero.
The tuning parameter controls the level of sparsity, where more vectors $\tilde{\alpha}^g_{j,k}$ are zero for higher $\lambda$.
In Appendix B.3, we establish consistency of the regularized score matching estimator assuming sparsity of $\tilde{\boldsymbol{\alpha}}^g_j$ and some additional regularity conditions.

As is the case for the group LASSO estimator, the regularized score matching estimator has an intractable limiting distribution because its bias and standard error diminish at the same rate.
We can obtain an asymptotically normal estimate by subtracting from the initial estimate an estimate of the bias.
In Appendix B.4, we construct such a bias-corrected estimate $\check{\alpha}^g_{j,k}$ that, for sufficiently large $n^g$, satisfies
\begin{align*}
\check{\alpha}^g_{j,k} \sim N\left(\alpha^{g,*}_{j,k}, \Omega^g_{j,k} \right),
\end{align*}
for a positive definite matrix $\Omega^g_{j,k}$.
Given bias-corrected estimates and a consistent estimate $\check{\Omega}^g_{j,k}$ of $\Omega^g_{j,k}$, we can test the null hypothesis \eqref{NullAdj} using the test statistic
\begin{align*}
S_{j,k} = \left( \check{\alpha}^{\mathrm{I}}_{j,k} - \check{\alpha}^{\mathrm{II}}_{j,k} \right)^\top \left(  \check{\Omega}^{\mathrm{I}}_{j,k}  + \check{\Omega}^{\mathrm{II}}_{j,k}  \right)^{-1}\left( \check{\alpha}^{\mathrm{I}}_{j,k} - \check{\alpha}^{\mathrm{II}}_{j,k} \right).
\end{align*}
Under the null hypothesis, the test statistic follows a chi-squared distribution with $d$ degrees of freedom.

\section{Numerical Studies}

In this section, we examine the performance of our proposed test in a simulation study.
We consider the neighborhood selection approach described in Section~3.
Our simulation study has three objectives: (1) to assess the stability of our estimators for the covariate-dependent networks, (2) to examine the effect of sample size on statistical power and type-I error control, and (3) to illustrate that failing to adjust for covariates can in some settings result in poor type-I error control or reduced statistical power.

\subsection{Implementation}

We first discuss implementation of the neighborhood selection approach.
The group LASSO estimate \eqref{GrpLASSO} does not exist in closed form, in contrast to the ordinary least squares estimate \eqref{AlphaOLS2}.
To solve \eqref{GrpLASSO}, we use the efficient algorithm implemented in the publicly available R package \texttt{gglasso} \cite{yang2015fast}.

The group LASSO estimator requires selection of a tuning parameter $\lambda$, which controls the sparsity of the estimate.
We select the tuning parameter  by performing $K$-fold cross-validation, using $K = 10$ folds.
Since the selection of $\lambda$ is sensitive to the scale of the columns of $\mathcal{V}_k^g$ in \eqref{VjLinear}, we scale the columns by their standard deviations prior to cross-validating.
After fitting the group LASSO with the selected tuning parameter, we convert the estimates back to their original scale by dividing the estimates by the standard deviations of the columns of $\mathcal{V}_k^g$.


\subsection{Simulation Setting}

In what follows, we describe our simulation setting.
In short, we generate data from the varying coefficient model \eqref{AdjModel1}, where we treat nodes $1$ through $(p-1)$ as predictors, and treat node $p$ as the response.
We first randomly generate data for nodes $1$ through $(p-1)$ in groups $\mathrm{I}$ and $\mathrm{II}$ from the same multivariate normal distribution.
We then construct $\eta^{g,*}_{j,k}$  and generate data for two covariates $W_i^g = (W_{i,1}^g, W_{i,2}^g)^\top$ so that one covariate acts as a confounding variable, and the other covariate should improve statistical power to detect differential associations after adjustment.

To simulate data for nodes $1$ through $(p-1)$, we first generate a random graph with $(p - 1)$ nodes and an edge density of .05 from a power law distribution with power parameter 5 \citep{newman2003structure}.
Denoting the edge set of the graph by $E$, we generate the $(p-1) \times (p-1)$ matrix $\Theta$ as
\[
\Theta_{j,k} = 
\begin{cases}
0 & (j,k) \notin E
\\
.5 & (j,k) \in E \text{ with 50\% probability}
\\
-.5 & (j,k) \in E \text{ with 50\% probability}
\end{cases},
\]
with $\Theta_{j,k} = \Theta_{k,j}$.
Defining by $a^*$ the smallest eigenvalue of $\Theta$, we set $\Sigma = (\Theta - (a^* - .1)I)^{-1}$, where $I$ is the identity matrix.
We then draw $(X_{i,1}^{g},\ldots,X_{i,p-1}^{g})^\top$ from a multivariate normal distribution with mean zero and covariance $\Sigma$ for $i = 1,\ldots,n^g$ for each group $g$.

We generate $W^{\mathrm{I}}_{i,1}$ from a $\text{Beta}(3/2, 1)$ distribution and $W^{\mathrm{II}}_{i,1}$ from a $\text{Beta}(1, 3/2)$ distribution.
We center and scale both $W^{\mathrm{I}}_{i,1}$ and $W^{\mathrm{II}}_{i,1}$ to the $(-1, 1)$ interval.
We generate $W^{\mathrm{I}}_{i,2}$ and $W_{i,2}^{\mathrm{II}}$ each from a Uniform$(-1, 1)$ distribution.

We consider two different choices for the varying coefficient functions $\eta^{g,*}_{j,k}$:
\begin{itemize}
\item
\textit{Linear Polynomial:}
\begin{align*}
&\eta^\mathrm{I,*}_{p,1}(w_1, w_2) = .5 + .5w_1; &&\eta^{\mathrm{II,*}}_{p,1}(w_1, w_2) = .5 + .5w_1
\\ \nonumber 
&\eta^\mathrm{I,*}_{p,2}(w_1, w_2) = .5 + .25w_2; &&\eta^{\mathrm{II,*}}_{p,2}(w_1, w_2) = .5 + .75w_2
\\ \nonumber 
&\eta^{\mathrm{I,*}}_{p,3}(w_1, w_2) = 0; &&\eta^{\mathrm{II,*}}_{p,3}(w_1, w_2) = .5,
\end{align*}
and $\eta^{g,*}_{p,k} = 0$ for $k \geq 4$.
\item
\textit{Cubic Polynomial:}
\begin{align*}
&\eta^\mathrm{I,*}_{p,1}(w_1, w_2) = .5 + .5\left(w_1 + w_1^2 + w_1^3\right); &&\eta^{\mathrm{II,*}}_{p,1}(w_1, w_2) = .5 + .5\left(w_1 + w_1^2 + w_1^3\right)
\\ \nonumber 
&\eta^\mathrm{I,*}_{p,2}(w_1, w_2) = .5 + .25\left(w_2 + w_2^3\right); &&\eta^{\mathrm{II,*}}_{p,2}(w_1, w_2) = .5 +  .75\left(w_2 + w_2^3\right)
\\ \nonumber 
&\eta^{\mathrm{I,*}}_{p,3}(w_1, w_2) = 0; &&\eta^{\mathrm{II,*}}_{p,3}(w_1, w_2) = .5,
\end{align*}
and $\eta^{g,*}_{p,k} = 0$ for $k \geq 4$.
\end{itemize}
The first covariate $W_{i,1}^g$ confounds the association between nodes $p$ and 1.
The distribution of $W_{i,1}^g$ depends on group membership, and $W_{i,1}^g$ affects the association between genes $p$ and $1$.
However, $\eta^{\mathrm{I},*}_{p,1}(w) = \eta^{\mathrm{II},*}_{p,1}(w)$ for all $w$.
Thus, $G^0_{p,1}$ in \eqref{NullAdj} holds while $H^0_{p,1}$ in \eqref{NullUnadj} fails, as depicted in Figure \ref{fig:ConfAndPower}a.
Failing to adjust for $W_1^g$ should therefore result in an inflated type-I error rate for the hypothesis $G^0_{p,1}$.
Adjusting for the second covariate $W_{i,2}^g$ should improve the power to detect the differential connection between nodes $p$ and $2$.
We have constructed $\eta^{g,*}_{p,2}$ so that $\mathds{E}\left[\eta^{\mathrm{I},*}\left(W^{\mathrm{I}}\right)\right] = \mathds{E}\left[\eta^{\mathrm{II},*}\left(W^{\mathrm{II}}\right)\right]$, though the association between nodes $p$ and 2 depends more strongly on $W^g$ in group $\mathrm{II}$ than in group $\mathrm{I}$.
Thus, $H^0_{p,2}$ holds while $G^0_{p,2}$ fails, as depicted in Figure \ref{fig:ConfAndPower}b. 
The association between nodes $p$ and $3$ does not depend on either covariate, though the association differs by group.
Thus, one should be able to identify a differential connection using either the adjusted or unadjusted test.
Node $p$ is conditionally independent of all other nodes in both groups.

For $i = 1,\ldots,n^g$, we generate $X^g_{i,p}$ as
\[
X^g_{i,p} =\sum_{k \neq p} \eta^g_{j,k}\left(W_i^g\right) X^g_{i,k} + \epsilon^g_i,
\]
where $\epsilon^g_i$ follows a normal distribution with zero mean and unit variance.
We use balanced sample sizes $n^{\mathrm{I}} = n^{\mathrm{II}} = n$ and consider $n \in \{80, 160, 240\}$.
We set the number of nodes $p = 40$. 
The graph for nodes $1$ through $(p-1)$ contains 15 edges.
Leaving $\Sigma$ fixed, we generate 400 random data sets following the above approach.

We consider two choices of the basis expansion $\phi$:
\begin{enumerate}
\item Linear basis: $\phi(w_1, w_2) = \begin{pmatrix}
1 & w_1 & w_2
\end{pmatrix}^\top$;
\item Cubic polynomial basis: $\phi(w_1, w_2) = \begin{pmatrix} 1 & w_1 & w_1^2 & w_1^3 & w_2 & w_2^2 & w_2^3 \end{pmatrix}^\top$.
\end{enumerate}
Using a linear basis, $d = 3$, and model \eqref{AdjModel1} has $117$ parameters.
With the cubic polynomial basis, $d = 7$, and there are $273$ parameters.

We compare our proposed methodology with the approach for differential network analysis without covariate adjustment described in Section 3.1. 
In the unadjusted analysis, ordinary least squares estimation is justified because although $(p-1)d$ is large with respect to $n$, $(p-1)$ is smaller than $n$.


\subsection{Simulation Results}

Figure \ref{fig:l2error} shows the Monte Carlo estimates of the expected $\ell_2$ error for the de-biased group LASSO estimates $\check{\alpha}^g_{p,k}$, $\mathds{E}\left[ \left\|d^{-1}\left(\check{\alpha}^{g}_{p,k} - \alpha^{g,*}_{p,k}\right)\right\|_2\right]$, for $k = 1,\ldots,(p-1)$.
We only report the $\ell_2$ error when the basis $\phi$ is correctly specified for the varying coefficient function $\eta^{g,*}_{p,k}$ --- that is, when $\phi$ is linear basis, and $\eta^{g,*}_{p,k}$ is a linear function or when $\phi$ is a cubic basis, and $\eta^{g,*}_{p,k}$ is a cubic function.
In both the linear and cubic polynomial settings, the average $\ell_2$ estimation error for $\alpha^{g,*}_{p,k}$ decreases with the sample size for all $k$, as expected.
We also find that in small samples, the estimation error is substantially lower in the linear setting than in the cubic setting.
This suggests that estimates are less stable in more complex models.

\begin{figure}[!t]
\centering
 \includegraphics[scale = 1]{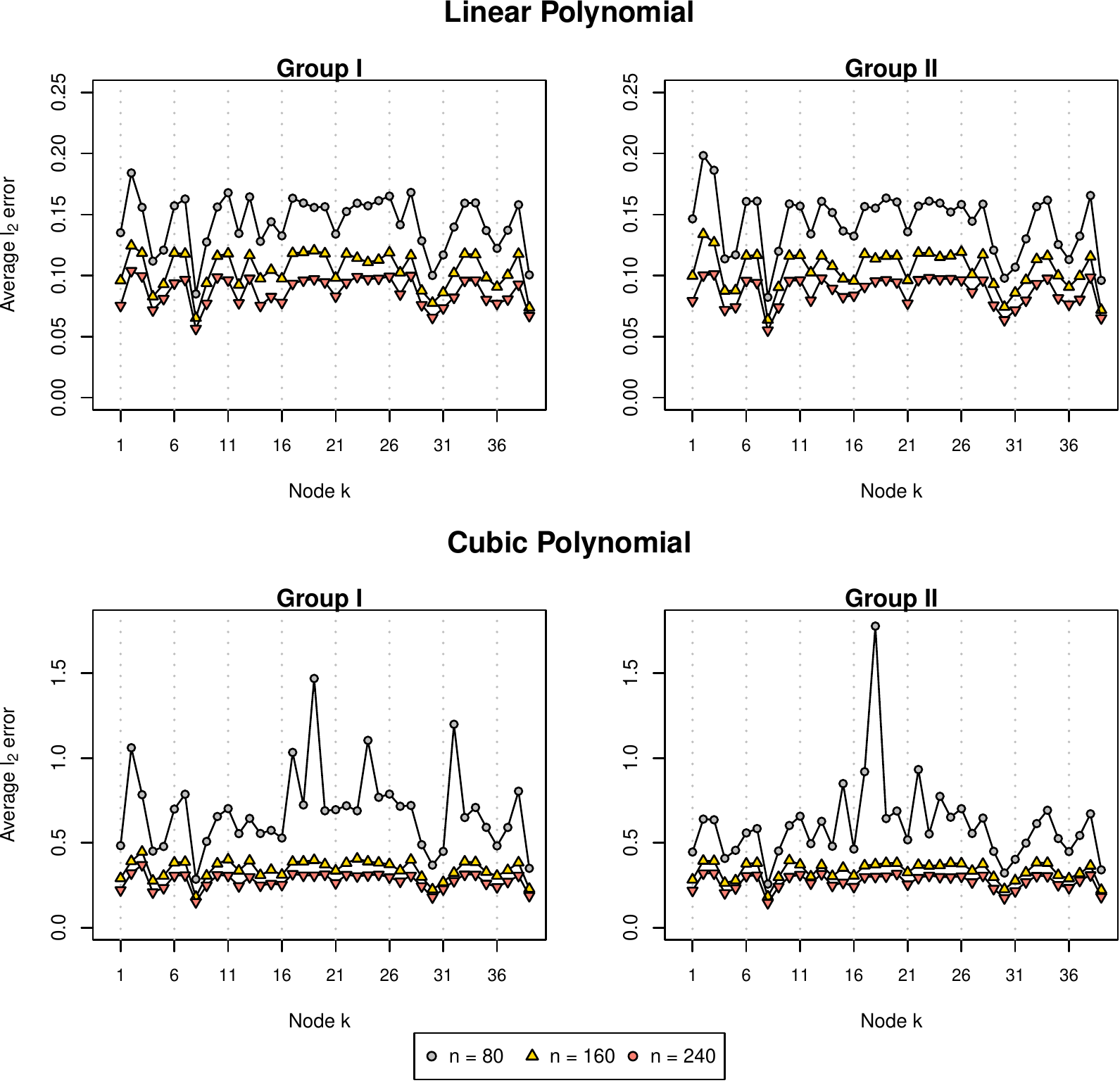}
\caption{Monte Carlo estimates of expected $\ell_2$ error, $\mathds{E}\left[\left\| d^{-1}\left(\check{\alpha}^g_{p,k} - \alpha^{g,*}_{p,k}\right)\right\|_2\right],$ for $k = 1,\ldots,39$. The linear polynomial plots display the $\ell_2$ error when $\eta^{g,*}_{j,k}$ is a linear function, and $\phi$ is a linear basis.  The cubic polynomial plots display the $\ell_2$ error when $\eta^{g,*}_{j,k}$ is a cubic polynomial, and $\phi$ is a cubic basis.}
\label{fig:l2error}
\end{figure}

In Table \ref{table:Sims1}, we report Monte Carlo estimates of the probability of rejecting $G^0_{p,k}$, the null hypothesis that nodes $p$ and $k$ are not differentially connected given $W^g$, for $k = 1$, $k = 2$, $k = 3$, and $k \geq 4$, using both the adjusted and unadjusted tests at the significance level $\kappa = .05$.
As the purpose of the simulation study is to examine the behavior of the edge-wise test, we do not perform a multiple testing correction.

For $k = 1$ (i.e., when $H^0_{p,k}$ fails, but $G^0_{p,k}$ holds), the unadjusted test is anti-conservative, and the probability of falsely rejecting $G^0_{p,k}$ increases with the sample size.
When an adjusted test is performed using a linear basis, and when $\eta^{g,*}_{p,1}$ is linear, the type-I error rate is slightly inflated but appears to approach the nominal level of $.05$ as the sample size increases.
However, when $\eta^{g,*}_{p,1}$ is a cubic function, and the linear basis is mis-specified, the type-I error rate is inflated, though it is still slightly lower than that of unadjusted test.
For both specifications of $\eta^{g,*}_{p,1}$, the covariate-adjusted test controls the type-I error rate near the nominal level when a cubic polynomial basis is used.
For $k = 2$, (i.e., when $H^0_{p,k}$ holds, but $G^0_{p,k}$ fails), the unadjusted test exhibits low power to detect differential associations.
The adjusted test provides greatly improved power when either a linear or cubic basis is used.
For $k = 3$, (i.e., when both $H^0_{p,k}$ and $G^0_{p,k}$ fail), the unadjusted test and both adjusted tests are well-powered against the null.
For $k \geq 4$ (i.e., when genes $p$ and $k$ are conditionally independent in both groups), the unadjusted test and the adjusted test with a linear basis both control the type-I error near the nominal level.
However, the covariate-adjusted test is conservative when a cubic basis is used.

\begin{table}[!h]
\centering
\small
\begin{tabular}{lllllllllll}
  \hline
  & & \multicolumn{3}{c}{Unadjusted} & \multicolumn{3}{c}{Linear Adjustment} & \multicolumn{3}{c}{Cubic Adjustment} \\
 & & $n = 80$ & $n = 160$ & $n = 240$ & $n = 80$ & $n = 160$ & $n = 240$ & $n = 80$ & $n = 160$ & $n = 240$ \\
 \hline
Linear $\eta^{g,*}_{p,k}$ & $k = 1$ & 0.15 & 0.278 & 0.385 & 0.13 & 0.09 & 0.072 & 0.04 & 0.062 & 0.05 \\ 
   & $k = 2$ & 0.042 & 0.078 & 0.045 & 0.27 & 0.532 & 0.73 & 0.08 & 0.27 & 0.52 \\ 
   & $k = 3$ & 0.48 & 0.912 & 0.988 & 0.605 & 0.922 & 0.965 & 0.218 & 0.738 & 0.902 \\ 
   & $k \geq 4$ & 0.052 & 0.054 & 0.053 & 0.045 & 0.048 & 0.048 & 0.009 & 0.017 & 0.025 \\ 
  Cubic $\eta^{g,*}_{p,k}$ & $k = 1$ & 0.21 & 0.505 & 0.668 & 0.358 & 0.315 & 0.342 & 0.07 & 0.055 & 0.07 \\ 
   & $k = 2$ & 0.052 & 0.068 & 0.082 & 0.6 & 0.882 & 0.975 & 0.195 & 0.73 & 0.93 \\ 
   & $k = 3$ & 0.408 & 0.84 & 0.978 & 0.55 & 0.898 & 0.982 & 0.202 & 0.772 & 0.945 \\ 
   & $k \geq 4$ & 0.056 & 0.05 & 0.054 & 0.053 & 0.054 & 0.052 & 0.009 & 0.02 & 0.027 \\ 
   \hline
\end{tabular}
\caption{Monte Carlo estimates of probability of rejecting $G^{0}_{p,k}$, the null hypothesis that nodes $p$ and $k$ are not differentially connected, given $W^g$. 
All tests are performed at the significance level $\kappa = .05$, and no multiple testing correction is performed.} 
\label{table:Sims1}
\end{table}

The simulation results corroborate our expectations and suggest that there are potential benefits to covariate adjustment.
We find that when the sample size is large, the covariate-adjusted test behaves reasonably well with either choice of basis function.
However, in small samples, the covariate-adjusted test is somewhat imprecise, and the type-I error rate can be slightly above or below the nominal level.
Practitioners should therefore exercise caution when using our proposed methodology in very small samples.

\section{Data Example}

Breast cancer classification based on expression of estrogen receptor hormone (ER) is prognostic of clinical outcomes.
Breast cancers can be classified as estrogen receptor positive (ER+) and estrogen receptor negative (ER-), with approximately 70\% of breast cancers being ER+ \citep{lumachi2013treatment}.
In ER+ breast cancer, the cancer cells require estrogen to grow; this has been shown to be associated with positive clinical outcomes, compared with ER- breast cancer \citep{carey2006race}.
Identifying differences between the biological pathways of ER+ and ER- breast cancers can be helpful for understanding the underlying disease mechanisms.

It is has been shown that age is associated with ER status and that age can be associated with gene expression \citep{khan1998estrogen, yang2015synchronized}.
This warrants consideration of age as an adjustment variable in a comparison of gene co-expression networks between ER groups.

We perform an age-adjusted differential analysis of the ER+ and ER- breast cancer networks, using publicly available data from The Cancer Genome Atlas (TCGA) \citep{weinstein2013cancer}.
We obtain clinical measurements and gene expression data from a total of 806 ER+ patients and 237 ER- patients.
We consider the set of $p = 145$ genes in the Kyoto Encyclopedia of Genes and Genomes (KEGG) breast cancer pathway \citep{kanehisa2000kegg}, and adjust for age as our only covariate.
The average age in the ER+ plus group is 59.3 years (SD = 13.3), and the average age in the ER- group is 55.9 years (SD = 12.4).
We use a linear basis for covariate adjustment.
In the ER+ group, the sample size is considerably larger than the number of the parameters, so we can fit the varying coefficient model \eqref{AdjModel1} using ordinary least squares.
We use the de-biased group LASSO to estimate the network for the ER- group because the sample size is smaller than the number of model parameters.
We compare the results from the covariate-adjusted analysis with the unadjusted approach described in Section 3.1.

To assess for differential connectivity between any two nodes $j$ and $k$, we can either treat node $j$ or node $k$ as the response in the varying coefficient model \eqref{AdjModel1}.
We can then test either of the hypotheses $G^0_{j,k}:\alpha^{\mathrm{I},*}_{j,k} = \alpha^{\mathrm{II},*}_{j,k}$ or $G^0_{k,j}:\alpha^{\mathrm{I},*}_{k,j} = \alpha^{\mathrm{II},*}_{k,j}$.
Our approach is to set our p-value for the test for differential connectivity between nodes $j$ and $k$ as the minimum of the p-values for the tests of $G^0_{j,k}$ and $G^0_{k,j}$, though we acknowledge that this strategy is anti-conservative.

Our objective is to identify all pairs of differentially connected genes, so we need to adjust for the fact that we perform a separate hypothesis test for each gene pair.
We account for multiplicity by controlling the false discovery rate at the level $\kappa = .05$ using the Benjamini-Yekutieli method \citep{benjamini2001control}.

\begin{figure}[!h]
\centering
\includegraphics[width=16.5cm]{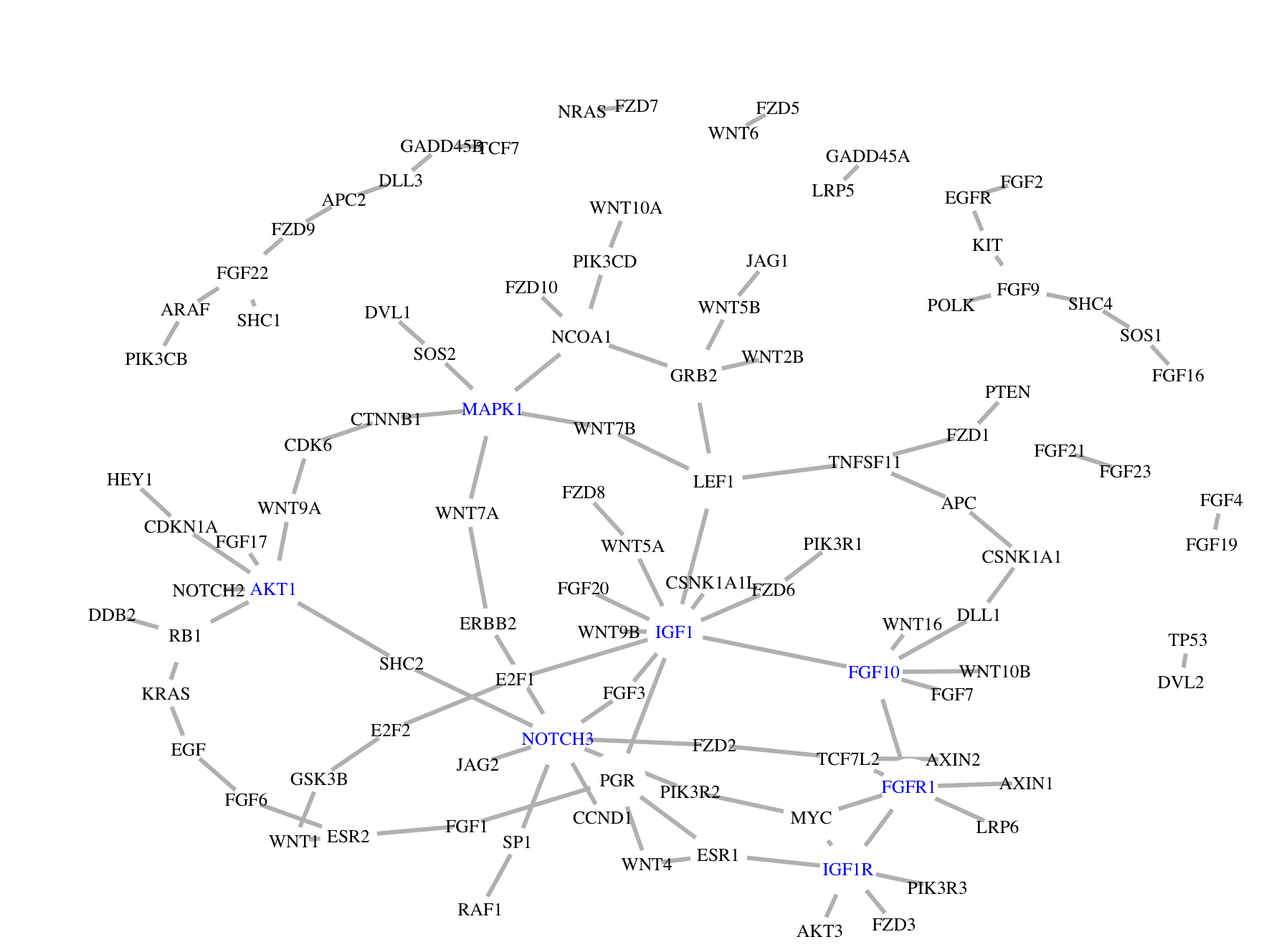}
\caption{Differential breast cancer network by estrogen receptor status from covariate-adjusted analysis.  Nodes with at least five differentially connected neighbors are colored blue. The false discovery rate is controlled at .05.
}
\label{fig:BCNetwork}
\end{figure}

The differential networks obtained from the unadjusted and adjusted analyses are substantially different.
We report 106 differentially connected edges from the adjusted analysis (shown in Figure \ref{fig:BCNetwork}), compared to only two such edges from the unadjusted analysis.
This suggests it is possible that relationship between the gene co-expression network and age differs by ER group.

\section{Discussion}

In this paper, we have addressed challenges that arise when performing differential network analysis \citep{shojaie2020differential} in the setting where the network depends on covariates.
Using both synthetic and real data, we showed that accounting for covariates can result in better control of type-I error and improved power.

We propose a parsimonious approach for covariate adjustment in differential network analysis.
A number of improvements and extensions can be made to our current work.
First, while this paper focuses on differential network analysis in exponential family models, our framework can be applied to other models where conditional dependence between any pair of nodes can be represented by a single scalar parameter.
This includes semi-parametric models such as the nonparanormal model \citep{liu2009nonparanormal}, as well as distributions defined over complex domains, which can be modeled using the generalized score matching framework \citep{yu2021}. 
Additionally, we only discuss testing edge-wise differences between the networks, though testing differences between sub-networks may also be of interest.
When the sub-networks are low-dimensional, one can construct a chi-squared test using similar test statistics as presented in Section 3 and Section 4 because joint asymptotic normality of a low-dimensional set of the estimators $\check{\alpha}^g_{j,k}$ can be readily established.
Such an approach is not applicable to high-dimensional sub-networks, but it may be possible to construct a calibrated test using recent results on simultaneous inference in high-dimensional models \citep{zhang2017simultaneous, yu2020simultaneous}. 
We can also improve the statistical efficiency of the network estimates by considering joint estimation procedures that borrow information across groups \citep{guo2011joint, danaher2014joint, saegusa2016joint}.
Finally, we assume that the relationship between the network and the covariates can be represented by a low-dimensional basis expansion.
Investigating nonparametric approaches that relax this assumption can be a fruitful area of research.

\section*{Funding}

The authors gratefully acknowledge the support of the NSF Graduate Research Fellowship Program under grant DGE-1762114 as well as NSF grant DMS-1561814 and NIH grant R01-GM114029.
Any opinions, findings, and conclusions or recommendations expressed in this material are those of the authors and do not necessarily reflect the views of the funding agencies.

%

\section*{Availability of Data}

This findings of this paper are supported by data from The Cancer Genome Atlas, which are accessible using the publicly available R package \texttt{RTCGA}.

\section*{Code Availability}

An implementation of the proposed methodology is available at \texttt{https://github.com/awhudson/CovDNA}.

\bibliographystyle{spmpsci}      
\bibliography{CovAdjGM-arxiv}   


\newpage

\section*{Appendix}

\subsection*{A \quad De-biased Group LASSO Estimator}

In this subsection, we derive a de-biased group LASSO estimator.
Our construction is essentially the same as the one presented in van de Geer \citep{van2016estimation}.

With $\mathcal{V}_j$ as defined in \eqref{VjLinear}, let $\mathcal{V}_{-j}^g = \left(\mathcal{V}^g_{1},\cdots,\mathcal{V}^g_{j-1}, \mathcal{V}^g_{j+1},\cdots, \mathcal{V}^g_p\right)$ be an $n \times (p-1)d$ dimensional matrix.
For $\alpha_{j,1}, \ldots, \alpha_{j,p} \in \mathds{R}^d$, let $\boldsymbol{\alpha}_j = \left(\alpha_{j,1}^\top, \cdots, \alpha_{j,p}^\top\right)^\top$, let $\mathcal{P}_j\left( \boldsymbol{\alpha}_j \right) = \sum_{k \neq j} \left\| \alpha_{j,k} \right\|_2$, and let $\nabla \mathcal{P}_j$ denote the sub-gradient of $\mathcal{P}_j$.
We can express the sub-gradient as $\nabla\mathcal{P}_j(\boldsymbol{\alpha}_j) = \left((\nabla \|\alpha_{j,1}\|_2)^\top, \cdots, (\nabla\|\alpha_{j,p}\|_2)^\top \right)^\top$ where  $\nabla \|{\alpha}_{j,k}\|_2 = \alpha_{j,k}/\|\alpha_{j,k}\|_2$ if $\|\alpha_{j,k}\|_2 \neq 0$, and $\nabla \|{\alpha}_{j,k}\|_2$ is otherwise a vector with $\ell_2$ norm less than one.
The KKT conditions for the group LASSO imply that the estimate $\tilde{\boldsymbol{\alpha}}^g_j$ satisfies
\begin{align*}
\left(n^g\right)^{-1}\left(\mathcal{V}_{-j}^g\right)^\top \left( \mathbf{X}_j^g - \mathcal{V}_{-j}^g \tilde{\boldsymbol{\alpha}}^g_j \right) = -\lambda \nabla \mathcal{P}_j\left(\tilde{\boldsymbol{\alpha}}^g_j \right).
\end{align*}
With some algebra, we can rewrite this as
\begin{align*}
\left(n^g\right)^{-1}\left(\mathcal{V}_{-j}^g\right)^\top \mathcal{V}_{-j}^g \left(\tilde{\boldsymbol{\alpha}}_j^g- \boldsymbol{\alpha}^{g,*}_j\right) = -\lambda \nabla \mathcal{P}_j\left(\tilde{\boldsymbol{\alpha}}^g_j \right) + \left(\mathcal{V}^g_{-j}\right)^\top \left(\mathbf{X}_j^g - \mathcal{V}_{-j}^{g} \boldsymbol{\alpha}^{g,*}_j \right).
\end{align*}
Let $\Sigma_j$ be defined as the matrix
\begin{align*}
\Sigma_j = \mathds{E}\left[\left(n^g\right)^{-1}\left(\mathcal{V}_{-j}^g\right)^\top \mathcal{V}_{-j}^g\right],
\end{align*}
and let $\tilde{M}_j$ be an estimate of $\Sigma_j^{-1}$. 
We can write $\left(\tilde{\boldsymbol{\alpha}}^g_j - \tilde{\boldsymbol{\alpha}}^{g,*}_j\right)$ as
\begin{align}
\left(\tilde{\boldsymbol{\alpha}}^g_j - \boldsymbol{\alpha}_j^{g,*} \right) = &\underset{\mathrm{(i)}}{\underbrace{-\lambda \tilde{M}_j\nabla \mathcal{P}_j\left(\tilde{\boldsymbol{\alpha}}^g_j \right)}} + \underset{\mathrm{(ii)}}{\underbrace{\left(n^g\right)^{-1}\tilde{M}_j\left(\mathcal{V}^g_{-j}\right)^\top \left(\mathbf{X}_j^g - \mathcal{V}_{-j}^{g} \boldsymbol{\alpha}^{g,*}_j \right)}} + \nonumber
\\
&\underset{\mathrm{(iii)}}{\underbrace{\left\{I - \left(n^g\right)^{-1}\tilde{M}_j\left(\mathcal{V}_{-j}^g\right)^\top \mathcal{V}_{-j}^g \right\} \left(\tilde{\boldsymbol{\alpha}}_j^g - \boldsymbol{\alpha}^{g,*}_j\right)}}.
\label{DeBiasExpLM}
\end{align}
The first term $(\mathrm{i})$ in \eqref{DeBiasExpLM} is an approximation for the bias of the group LASSO estimate.
This term is a function only of the observed data and not of any unknown quantities.
This term can  therefore be directly added to the initial estimate $\tilde{\boldsymbol{\alpha}}_j^g$.
If $\tilde{M}_j$ is a consistent estimate of $\Sigma_j^{-1}$, the second term $(\mathrm{ii})$ is asymptotically equivalent to
\begin{align*}
\Sigma^{-1}_j
\left(\mathcal{V}^g_{-j}\right)^\top \left(\mathbf{X}_j^g - \mathcal{V}_{-j}^{g} \boldsymbol{\alpha}^{g,*}_j \right).
\end{align*}
Thus, $(\mathrm{ii})$ is asymptotically equivalent to a sample average of mean zero \textit{i.i.d.} random variables.
The central limit theorem can then be applied to establish convergence in distribution to the multivariate normal distribution at an $n^{1/2}$  rate for any low-dimensional sub-vector.
The third term will also be asymptotically negligible if $\tilde{M}_j$ is an approximate inverse of $(n^g)^{-1}\left(\mathcal{V}_{-j}^g\right)^\top\mathcal{V}^g_{-j}$.
This would suggest that an estimator of the form
\begin{align*}
\check{\boldsymbol{\alpha}}_j^g = \tilde{\boldsymbol{\alpha}}_j^g + \lambda \tilde{M}_j \nabla \mathcal{P}_j\left( \tilde{\boldsymbol{\alpha}}^g_j \right)
\end{align*}
will be asymptotically normal for an appropriate choice of $\tilde{M}_j$.

Before describing our construction of $\tilde{M}_j$, we find it helpful to consider an alternative expression for $\Sigma^{-1}_j$.
We define the $d \times d$ matrices $\Gamma^*_{j,k,l}$ as
\begin{align}
&\Gamma^*_{j,k,1},\ldots,\Gamma^*_{j,k,p}  =  \underset{\Gamma_1,\ldots,\Gamma_p \in \mathds{R}^{d\times d}}{\text{arg min}}
\mathds{E}\left[
\text{trace}\left\{ 
\left(n^g\right)^{-1}
\left(\mathcal{V}^g_k - \sum_{l \neq k,j} \mathcal{V}^g_{l} \Gamma_l \right)^\top
\left(\mathcal{V}^g_k - \sum_{l \neq k,j} \mathcal{V}^g_{l} \Gamma_l \right)
\right\}\right].
\label{Ugh}
\end{align}
We also define the $d \times d$ matrix $\tilde{C}_{j,k}$ as
\begin{align*}
C^*_{j,k} = \mathds{E}\left[\left(n^g\right)^{-1} \left(\mathcal{V}_k^g -
  \sum_{l \neq k,j} \mathcal{V}^g_{l}\Gamma^*_{j,k,l}\right)^\top\mathcal{V}_k^g\right].
\end{align*}
It can be shown that $\Sigma^{-1}_j$ can be expressed as
\begin{align*}
\Sigma^{-1}_j =
\begin{pmatrix}
\left(C_{j,1}^*\right)^{-1} & \cdots & \mathbf{0}
\\
\vdots & \ddots & \vdots
\\
\mathbf{0} & \cdots & \left(C_{j,p}^*\right)^{-1}
\end{pmatrix}
\begin{pmatrix}
I & -\Gamma^{*}_{j,1,2} & \cdots & -\Gamma^*_{j,1,p} 
\\
-\Gamma^*_{j,2,1} & I & \cdots & -\Gamma^*_{j,2,p}
\\
\vdots & \vdots & \ddots & \vdots
\\
-\Gamma^*_{j,p,1} & -\Gamma^*_{j,p,2} & \cdots & I
\end{pmatrix}
.
\end{align*}
We can thus estimate $\Sigma_j^{-1}$ by performing a series of regressions to estimate each matrix $\Gamma^*_{j,k,l}$.

Following the approach of van de Geer et al. \cite{van2014asymptotically}, we use a group LASSO variant of the nodewise LASSO to construct $\tilde{M}_j$.
To proceed, we require some additional notation.
For any $d \times d$ matrix $\Gamma = (\gamma_1, \cdots, \gamma_d)$ for $d-$dimensional vectors $\gamma_c$, let $\|\Gamma \|_{2,*} = \sum_{c = 1}^d \|\gamma_c\|_2$.
Let $\nabla\| \Gamma \|_{2,*} = \left( \gamma_1/\|\gamma_1\|_2,\cdots,\gamma_d/ \|\gamma_d\|_2 \right)$ be the subgradient of $\|\Gamma \|_{2,*}$.
We use the group LASSO to obtain estimates $\tilde{\Gamma}_{j,k,l}$ of $\Gamma^*_{j,k,l}$:
\begin{align}
&\tilde{\Gamma}_{j,k,1},\ldots,\tilde{\Gamma}_{j,k,p}  =  \nonumber
\\
&\underset{\Gamma_1,\ldots,\Gamma_p \in \mathds{R}^{d\times d}}{\text{arg min}}
\text{trace}\left\{ 
\left(n^g\right)^{-1}
\left(\mathcal{V}^g_k - \sum_{l \neq k,j} \mathcal{V}^g_{l} \Gamma_l \right)^\top
\left(\mathcal{V}^g_k - \sum_{l \neq k,j} \mathcal{V}^g_{l} \Gamma_l \right)
\right\}
 +
 \omega \sum_{l \neq k,j} \|\Gamma_l \|_{2,*}.
\label{Ugh}
\end{align}
We then estimate $C^*_{j,k}$ as
\begin{align*}
\tilde{C}_{j,k} = \left(n^g\right)^{-1}
 \left(\mathcal{V}^g_k - \sum_{l \neq k,j} \mathcal{V}^g_{l} \tilde{\Gamma}_{j,k,l} \right)^\top\left(\mathcal{V}_k^g\right).
\end{align*}
Our estimate $\tilde{M}_j$ takes the form
\begin{align*}
\tilde{M}_j =
\begin{pmatrix}
\tilde{C}^{-1}_{j,1} & \cdots & \mathbf{0}
\\
\vdots & \ddots & \vdots
\\
\mathbf{0} & \cdots & \tilde{C}^{-1}_{j,p} 
\end{pmatrix}
\begin{pmatrix}
I & -\tilde{\Gamma}_{j,1,2} & \cdots & -\tilde{\Gamma}_{j,1,p} 
\\
-\tilde{\Gamma}_{j,2,1} & I & \cdots & -\tilde{\Gamma}_{j,2,p}
\\
\vdots & \vdots & \ddots & \vdots
\\
-\tilde{\Gamma}_{j,p,1} & -\tilde{\Gamma}_{j,p,2} & \cdots & I
\end{pmatrix}
.
\end{align*}

With this construction of $\tilde{M}_j$, we can establish a bound on the remainder term $(\mathrm{iii})$ in \eqref{DeBiasExpLM}.
To show this, we make use of the following lemma, which states a special case of the dual norm inequality for the group LASSO norm $\mathcal{P}_j$ (see, e.g., Chapter 6 of van de Geer \citep{van2016estimation}).
\begin{lemma}
Let $a_1,\ldots,a_p$ and $b_1,\ldots,b_p$ be $d$-dimensional vectors, and let $\mathbf{a} = \left(a_1^\top,\cdots,a_p^\top\right)^\top$ and $\mathbf{b} = \left(b_1^\top,\cdots,b_p^\top\right)^\top$ be $pd$-dimensional vectors.
Then
\begin{align*}
\langle \mathbf{a}, \mathbf{b}\rangle \leq 
\left(\sum_{j=1}^p \|a_j\|_2 \right) \max_{j} \left\| b_j  \right\|_2.
\end{align*}
\end{lemma}
\noindent The KKT conditions for \eqref{Ugh} imply that for all $l \neq j,k$
\begin{align}
\left(n^g\right)^{-1}\left(\mathcal{V}^g_l\right)^\top\left(\mathcal{V}^g_k - \sum_{r \neq k,j} \mathcal{V}^g_{r} \tilde{\Gamma}_{j,k,r}\right) = -\omega \nabla \left\| \tilde{\Gamma}_{j,k,l} \right\|_{2,*}.
\label{NodewiseKKT}
\end{align}
\noindent Lemma 1 and \eqref{NodewiseKKT} imply that
\begin{align*}
\left\| \begin{pmatrix}
\tilde{C}_{j,1} & \cdots & \mathbf{0}
\\
\vdots & \ddots & \vdots
\\
\mathbf{0} & \cdots & \tilde{C}_{j,p} 
\end{pmatrix}
\left\{I - \left(n^g\right)^{-1}\tilde{M}_j\left(\mathcal{V}_{-j}^g\right)^\top \mathcal{V}_{-j}^g \right\} \left(\tilde{\boldsymbol{\alpha}}_j^g - \boldsymbol{\alpha}^{g,*}_j\right)
 \right\|_{\infty} 
 \leq 
 \omega \mathcal{P}_{j}\left(\tilde{\boldsymbol{\alpha}}_j^g - \boldsymbol{\alpha}^{g,*}_j\right),
\end{align*}
where $\|\cdot\|_{\infty}$ is the $\ell_\infty$ norm.
With $\omega \asymp \left\{\log(p)/n\right\}^{1/2}$, $\tilde{M}_j$ can be shown to be consistent under sparsity of $\Gamma^{*}_{j,k,l}$ (i.e., only a few matrices $\Gamma^*_{j,k,l}$ have some nonzero columns) and some additional regularity conditions.
Additionally, it can be shown under sparsity of $\boldsymbol{\alpha}^{g,*}$ (i.e., very few vectors $\alpha^{g,*}_{j,k}$ are nonzero) and some additional regularity conditions that $\mathcal{P}_j\left(\tilde{\boldsymbol{\alpha}}_j^g - \boldsymbol{\alpha}_j^{g,*} \right) = O_P\left(\left\{\log(p)/n \right\}^{1/2}\right)$.
Thus, a scaled version of the remainder term ($\textrm{iii}$) is $o_P(n^{-1/2})$ if $n^{-1/2}\log(p) \to 0$.
We refer readers to Chapter 8 of B\"uhlmann and van de Geer \citep{buhlmann2011statistics} for a more comprehensive discussion of assumptions required for consistency of the group LASSO.

We now express the de-biased group LASSO estimator for $\alpha^{g,*}_{j,k}$ as
\begin{align}
\check{\alpha}^g_{j,k} = 
\tilde{\alpha}^g_{j,k} + 
\left(n^g\right)^{-1} 
\tilde{C}^{-1}_{j,k}
\left(
 \mathcal{V}^g_{k} -
 \sum_{l \neq j, k} \tilde{\Gamma}_{j,k,l} \mathcal{V}_{l}^g
 \right)^\top
 \left( \mathbf{X}^g_j - \mathcal{V}^g_{-j} \tilde{\boldsymbol{\alpha}}^g_j \right).
\end{align}
We have established that $\check{\alpha}^g_{j,k}$ can be written as
\begin{align*}
\tilde{C}_{j,k}
\left(\check{\alpha}^g_{j,k} - \alpha^{g,*}_{j,k}\right) = 
\left(n^g\right)^{-1}
\left(
\mathcal{V}^g_{k} -
\sum_{l \neq j, k} \Gamma^*_{j,k,l} \mathcal{V}_{l}^g
\right)^\top
\left( \mathbf{X}^g_j - \mathcal{V}^g_{-j} \boldsymbol{\alpha}^{g,*}_j \right) + o_P(n^{-1/2}).
\end{align*}
As stated above, the central limit theorem implies asymptotic normality of  $\check{\alpha}^g_{j,k}$.

We now construct an estimate for the variance  of  $\check{\alpha}^g_{j,k}$.
Suppose the residual $\mathbf{X}^g_j - \mathcal{V}^g_{-j} \boldsymbol{\alpha}^{g,*}_j$ is independent of $\mathcal{V}^g$, and let $\tau_j^g$ denote the residual variance 
\begin{align*}
\tau_j^g = \mathds{E}\left[\left(X_{j}^g - \sum_{k \neq j} \left\langle \phi_i^g, \alpha_{j,k}^{g,*} \right\rangle X_{k}^g\right)^2\right].
\end{align*}
We can approximate the variance of $\check{\alpha}^g_{j,k}$  as
\begin{align}
\check{\Omega}^g_{j,k} = \left(n^g\right)^{-2}\tau_j^g 
\tilde{C}^{-1}_{j,k}
\left(
 \mathcal{V}^g_{k} -
 \sum_{l \neq j, k} \tilde{\Gamma}_{j,k,l} \mathcal{V}_{l}^g
 \right)^\top
\left(
 \mathcal{V}^g_{k} -
 \sum_{l \neq j, k} \tilde{\Gamma}_{j,k,l} \mathcal{V}_{l}^g
 \right)
 \left(\tilde{C}^{-1}_{j,k}\right)^{\top}.
\end{align}
As $\tau_j^g$ is typically unknown, we instead us the estimate
\begin{align*}
\tilde{\tau}_j^g = \frac{\left\| \mathbf{X}^g_j - \mathcal{V}^g_{-j} \tilde{\boldsymbol{\alpha}}^g_j  \right\|_2^2}{n - \widehat{df}},
\end{align*}
where $\widehat{df}$ is an estimate of the degrees of freedom for the group LASSO estimate $\tilde{\boldsymbol{\alpha}}_j^g$.
In our implementation, we use the estimate proposed by Breheny and Huang \citep{breheny2009penalized}.

%

\subsection*{B \quad Generalized Score Matching Estimator}

In this section, we establish consistency of the  regularized score matching estimator and derive a bias-corrected estimator.

\subsubsection*{B.1 \quad Form of Generalized Score Matching Loss}

Below, we restate Theorem 3 of \cite{yu2019generalized}, which provides conditions under which the score matching loss \eqref{GenScoreLoss1} can be expressed as \eqref{GenScoreLoss2}.

\begin{theorem}
Assume the following conditions hold:
\begin{align*}
&\lim_{z_j \to \infty} h^*(z) (z_j) \left\{\frac{\partial}{\partial z_j} h(z_j) \right\} = 0 \quad \forall \, z_1,\ldots,z_{j-1},z_{j+1},\ldots,z_p \in \mathds{R}_+, \quad \forall h \in \mathcal{H}
\\
&\lim_{z_j \to 0} h^*(z) (z_j) \left\{\frac{\partial}{\partial z_j} h(z_j) \right\} = 0 \quad \forall \, z_1,\ldots,z_{j-1},z_{j+1},\ldots,z_p \in \mathds{R}_+, \quad \forall h \in \mathcal{H}
\\
&\sup_{h \in \mathcal{H}} \int \left\| \nabla \log h(z) \circ v^{1/2}(z) \right\|_2^2 h^*(z)dz < \infty
\\
&\sup_{h \in \mathcal{H}} \int \left\| \left\{\nabla \log h(z) \circ v^{1/2}(z)\right\}' \right\|_1 h^*(z)dz < \infty,
\end{align*}
where the prime symbol denotes the element-wise derivative.
Then \eqref{GenScoreLoss1} and \eqref{GenScoreLoss2} are equivalent up to an additive constant that does not depend on $h$.
\end{theorem}

\subsubsection*{B.2 \quad Generalized Score Matching Estimator in Low Dimensions}

In this section, we provide an explicit form for the generalized score matching estimator in the low-dimensional setting and state its limiting distribution.
We first introduce some additional notation below that allows for the generalized score matching loss to be written in a condensed form.
Recall the form of the conditional density for the pairwise interaction model in \eqref{CondExpFam}.
We define
\begin{align*}
&\mathcal{V}^g_{j,k,1} = 
\begin{pmatrix}
v_j^{1/2}\left(X^g_{1,j}\right)\dot{\psi}\left(X^g_{1,j}, X^g_{1,k}\right) \times \phi\left(W^g_{1}\right)
\\
\vdots
\\
v_j^{1/2}\left(X^g_{n^g,j}\right) \dot{\psi}\left(X^g_{n^g,j}, X^g_{n^g,k}\right) \times \phi\left(W^g_{n^g}\right)
\end{pmatrix},
\\ \\
&\mathcal{V}^g_{2,j} = 
\begin{pmatrix}
v_j^{1/2}\left(X^g_{1,j}\right) \times \left\{ \dot{\zeta}\left(X^g_{1,j}, \phi_1(W^g_{1})\right),\cdots,\dot{\zeta}\left(X^g_{1,j}, \phi_d(W^g_{1})\right) \right\}
\\
\vdots
\\
v_j^{1/2}\left(X^g_{n^g,j}\right) \times \left\{ \dot{\zeta}\left(X^g_{n^g,j}, \phi_1(W^g_{n^g})\right),\cdots,\dot{\zeta}\left(X^g_{n^g,j}, \phi_d(W^g_{n^g})\right) \right\}
\end{pmatrix},
\end{align*}
\begin{align*}
&\mathcal{U}^g_{j,k,1} 
= 
\begin{pmatrix}
\left\{\dot{v}_j\left(X^g_{1,j}\right)\dot{\psi}\left(X^g_{1,j}, X^g_{1,k}\right)  + v_j\left(X^g_{1,j}\right)\ddot{\psi}\left(X^g_{1,j}, X^g_{1,k}\right)  \right\} \times \phi\left(W^g_{1}\right)
\\
\vdots
\\
\left\{\dot{v}_j\left(X^g_{1,j}\right)\dot{\psi}\left(X^g_{n^g,j}, X^g_{n^g,k}\right)  + v_j\left(X^g_{n^g,j}\right)\ddot{\psi}\left(X^g_{1,j}, X^g_{n^g,k}\right)  \right\} \times \phi\left(W^g_{n^g}\right)
\end{pmatrix},
\\ \\
&\mathcal{U}^g_{j,2} = 
\begin{pmatrix}
v_j\left(X_{1,j}^g\right) \ddot{\zeta}\left(X^g_{1,j}, \phi_1(W^g_{1})\right)
&
\cdots 
&
v_j\left(X_{1,j}^g\right) \ddot{\zeta}\left(X^g_{1,j}, \phi_d(W^g_{1})\right)
\\
\vdots & \ddots & \vdots
\\
v_j\left(X_{n^g,j}^g\right) \ddot{\zeta}\left(X^g_{n^g,j}, \phi_1(W^g_{n^g})\right)
&
\cdots 
&
v_j\left(X_{n^g,j}^g\right) \ddot{\zeta}\left(X^g_{n^g,j}, \phi_d(W^g_{n^g})\right)
\end{pmatrix} +
\\
&\quad\quad\quad
\begin{pmatrix}
\dot{v}_j\left(X_{1,j}^g\right) \dot{\zeta}\left(X^g_{1,j}, \phi_1(W^g_{1})\right) & \cdots & \dot{v}_j\left(X_{1,j}^g\right) \dot{\zeta}\left(X^g_{1,j}, \phi_d(W^g_{1})\right)
\\
\vdots & \ddots & \vdots
\\
\dot{v}_j\left(X_{n^g,j}^g\right) \dot{\zeta}\left(X^g_{n^g,j}, \phi_1(W^g_{n^g})\right) & \cdots & \dot{v}_j\left(X_{n^g,j}^g\right) \dot{\zeta}\left(X^g_{n^g,j}, \phi_d(W^g_{n^g})\right)
\end{pmatrix},
\\  \\
&\mathcal{V}^g_{j,1} =
\begin{pmatrix}
\mathcal{V}^g_{j,1,1} \\ \vdots \\ \mathcal{V}^g_{j,p,1}
\end{pmatrix}; 
\quad
\mathcal{U}^g_{j,1} =
\begin{pmatrix}
\mathcal{U}^g_{1,j,1} \\ \vdots \\ \mathcal{U}^g_{j,p,1}
\end{pmatrix}.
\end{align*}

Let $\boldsymbol{\alpha}_j = \left(\alpha_{j,1}^\top, \cdots,\alpha_{j,p}^\top\right)^\top$ for $\alpha_{j,k} \in \mathds{R}^d$ and $\boldsymbol{\theta}_j = (\theta_{j,1}, \cdots, \theta_{j,d})^\top$ for $\theta_{j,c} \in \mathds{R}$.
We can express the empirical score matching loss \eqref{ExpFamScoreLoss} as
\begin{align*}
L^g_{n,j}(\boldsymbol{\alpha}_j, \boldsymbol{\theta}_j) &=
 \left(2n^g\right)^{-1} \left( \mathcal{V}_{j,1}^g \boldsymbol{\alpha}_j + \mathcal{V}^g_{2,j} \boldsymbol{\theta}_j  \right)^\top \left( \mathcal{V}_{j,1}^g \boldsymbol{\alpha}_j+ \mathcal{V}^g_{2,j} \boldsymbol{\theta}_j \right) + \left(n^g\right)^{-1}\mathbf{1}^\top \left( \mathcal{U}^g_{1,j} \boldsymbol{\alpha}_j  + \mathcal{U}^g_{2,j} \boldsymbol{\theta}_j  \right).
\end{align*}
We write the gradient of the risk function as
\begin{align*}
\nabla L^g_{n,j}(\boldsymbol{\alpha}_j, \boldsymbol{\theta}_j) &= 
\left(n^g\right)^{-1}
\begin{pmatrix}
\left(\mathcal{V}_{j,1}^g\right)^\top\mathcal{V}_{j,1}^g  & \left(\mathcal{V}_{j,1}^g\right)^\top\mathcal{V}_{j,2}^g
\\
\left(\mathcal{V}_{j,2}^g\right)^\top\mathcal{V}_{j,1}^g & \left(\mathcal{V}_{j,2}^g\right)^\top\mathcal{V}_{j,2}^g
\end{pmatrix}
\begin{pmatrix}
\boldsymbol{\alpha}_j
\\
\boldsymbol{\theta}_j
\end{pmatrix}
+
\left(n^g\right)^{-1}
\begin{pmatrix}
\left(\mathcal{U}_{j,1}^g\right)^\top\mathbf{1}
\\
\left(\mathcal{U}_{j,2}^g\right)^\top\mathbf{1}
\end{pmatrix}.
\end{align*}
Thus, the minimizer $(\hat{\boldsymbol{\alpha}}^g_j, \hat{\boldsymbol{\theta}}^g_j)$ of the empirical loss takes the form
\begin{align*}
\begin{pmatrix}
\hat{\boldsymbol{\alpha}}^g_j
\\
\hat{\boldsymbol{\theta}}^g_j
\end{pmatrix} 
=
-
\begin{pmatrix}
\left(\mathcal{V}_{j,1}^g\right)^\top\mathcal{V}_{j,1}^g  & \left(\mathcal{V}_{j,1}^g\right)^\top\mathcal{V}_{j,2}^g
\\
\left(\mathcal{V}_{j,2}^g\right)^\top\mathcal{V}_{j,1}^g & \left(\mathcal{V}_{j,2}^g\right)^\top\mathcal{V}_{j,2}^g
\end{pmatrix}^{-1}
\begin{pmatrix}
\left(\mathcal{U}_{j,1}^g\right)^\top\mathbf{1}
\\
\left(\mathcal{U}_{j,2}^g\right)^\top\mathbf{1}
\end{pmatrix}.
\end{align*}
By applying Theorem 5.23 of van der Vaart \cite{van2000asymptotic}, 
\begin{align*}
\left(n^g\right)^{1/2}
\begin{pmatrix}
\hat{\boldsymbol{\alpha}}^g_j - \boldsymbol{\alpha}_j^{g,*}
\\
\hat{\boldsymbol{\theta}}^g_j - \boldsymbol{\theta}_j^{g,*}
\end{pmatrix} 
\to_d
N\left(
0,
\begin{pmatrix}
A B A
\end{pmatrix}
\right),
\end{align*}
where the matrices $A$ and $B$ are defined as
\begin{align*}
&A = 
\mathds{E}
\left[
\left(n^g\right)^{-1}
\begin{pmatrix}
\left(\mathcal{V}_{j,1}^g\right)^\top\mathcal{V}_{j,1}^g  & \left(\mathcal{V}_{j,1}^g\right)^\top\mathcal{V}_{j,2}^g
\\
\left(\mathcal{V}_{j,2}^g\right)^\top\mathcal{V}_{j,1}^g & \left(\mathcal{V}_{j,2}^g\right)^\top\mathcal{V}_{j,2}^g
\end{pmatrix}\right]^{-1},
\\
&B = 
\text{Cov}
\left(
\left(n^g\right)^{-1}
\begin{pmatrix}
\left(\mathcal{V}_{j,1}^g\right)^\top\mathcal{V}_{j,1}^g  & \left(\mathcal{V}_{j,1}^g\right)^\top\mathcal{V}_{j,2}^g
\\
\left(\mathcal{V}_{j,2}^g\right)^\top\mathcal{V}_{j,1}^g & \left(\mathcal{V}_{j,2}^g\right)^\top\mathcal{V}_{j,2}^g
\end{pmatrix}
\begin{pmatrix}
\boldsymbol{\alpha}^{g,*}_j
\\
\boldsymbol{\theta}^{g,*}_j
\end{pmatrix}
+
\left(n^g\right)^{-1}
\begin{pmatrix}
\left(\mathcal{U}_{j,1}^g\right)^\top\mathbf{1}
\\
\left(\mathcal{U}_{j,2}^g\right)^\top\mathbf{1}
\end{pmatrix}
\right).
\end{align*}
We estimate the variance of $(\hat{\boldsymbol{\alpha}}^g_j, \hat{\boldsymbol{\theta}}^g_j)$ as $\hat{\Omega}^g_j = \left(n^g\right)^{-1}\hat{A} \hat{B} \hat{A}$, where
\begin{align*}
&\hat{A} =
n^g
\begin{pmatrix}
\left(\mathcal{V}_{j,1}^g\right)^\top\mathcal{V}_{j,1}^g  & \left(\mathcal{V}_{j,1}^g\right)^\top\mathcal{V}_{j,2}^g
\\
\left(\mathcal{V}_{j,2}^g\right)^\top\mathcal{V}_{j,1}^g & \left(\mathcal{V}_{j,2}^g\right)^\top\mathcal{V}_{j,2}^g
\end{pmatrix}^{-1}, 
\\
&\hat{B} = \left(n^g\right)^{-1}\hat{\xi}^\top\hat{\xi},
\quad
\hat{\xi} =
\begin{pmatrix}
\text{diag}\left(\mathcal{V}_{j,1}^g\hat{\boldsymbol{\alpha}}^g_j  + \mathcal{V}_{j,2}^g \hat{\boldsymbol{\theta}}^g_j \right)\mathcal{V}_{j,1}^g
\\
\text{diag}\left(\mathcal{V}_{j,1}^g\hat{\boldsymbol{\alpha}}^g_j  + \mathcal{V}_{j,2}^g \hat{\boldsymbol{\theta}}^g_j \right) \mathcal{V}_{j,2}^g
\end{pmatrix}
+
\begin{pmatrix}
\mathcal{U}_{j,1}^g
\\
\mathcal{U}_{j,2}^g
\end{pmatrix}.
\end{align*}

\subsubsection*{B.3 \quad Consistency of Regularized Generalized Score Matching Estimator}

In this subsection, we argue that the regularized generalized score matching estimators $\tilde{\boldsymbol{\alpha}}^g_j$ and $\tilde{\boldsymbol{\theta}}^g_j$ from \eqref{RegularizedScoreMatching} are consistent.
Let $\mathcal{P}_j(\boldsymbol{\alpha}_j) = \sum_{j=1}^p \|\alpha_{j,k}\|_2$.
We establish convergence rates of $\mathcal{P}_j\left( \tilde{\boldsymbol{\alpha}}_j^{g} - \boldsymbol{\alpha}_j^{g,*} \right)$ and $\left\|\tilde{\boldsymbol{\theta}}^g_j - \boldsymbol{\theta}_j^{g,*} \right\|_2$.
Our approach is based on proof techniques described in B\"uhlmann and van de Geer \cite{buhlmann2011statistics}.

Our result requires a notion of compatibility between the penalty function $\mathcal{P}_j$ and  the loss $L^g_{n,j}$.
Such notions are commonly assumed in the high-dimensional literature.
Below, we define the compatibility condition.
\begin{definition}[Compatibility Condition]
Let $S$ be a set containing  indices of the nonzero elements of $\boldsymbol{\alpha}_j^{g,*}$, and let $\bar{S}$ denote the complement of $S$.
Let $\mathds{1}_{S}$ be a $(p-1)d$-dimensional vector where the $r$-th element is one if $r \in S$, and zero otherwise.
The group LASSO compatibility condition holds for the index set $S \subset \{1,\ldots,p\}$ and for constant $C > 0$ if for all $\Omega\left(\boldsymbol{\alpha}_j \circ \mathds{1}_S\right) \leq 3\Omega\left(\boldsymbol{\alpha}_j \circ \mathds{1}_{\bar{S}} \right) + \| \boldsymbol{\theta}_j \|_2$,
\begin{align*}
\Omega\left( \boldsymbol{\alpha}_j \circ \mathds{1}_{S} \right) + \| \boldsymbol{\theta}_j \|/2 \leq \frac{|S|^{1/2}}{C} 
\left\{
\begin{pmatrix}
\boldsymbol{\alpha}_j^\top
&
\boldsymbol{\theta}_j^\top
\end{pmatrix}
\begin{pmatrix}
\left(\mathcal{V}_{j,1}^g\right)^\top\mathcal{V}_{j,1}^g  & \left(\mathcal{V}_{j,1}^g\right)^\top\mathcal{V}_{j,2}^g
\\
\left(\mathcal{V}_{j,2}^g\right)^\top\mathcal{V}_{j,1}^g & \left(\mathcal{V}_{j,2}^g\right)^\top\mathcal{V}_{j,2}^g
\end{pmatrix}
\begin{pmatrix}
\boldsymbol{\alpha}_j 
\\
\boldsymbol{\theta}_j
\end{pmatrix} 
\right\}^{1/2},
\end{align*} 
where $\circ$ is the element-wise product operator.
\end{definition}

\begin{theorem}
Let $\mathcal{E}$ be the set
\begin{align*}
\mathcal{E} = 
&\left\{
\max_{k \neq j}
\left\{
\left\| 
\left(\mathcal{V}_{j,k,1}^g\right)^\top \left(\mathcal{V}_{j,1}^g \boldsymbol{\alpha}_j^{g,*} + \mathcal{V}_{j,2}^g\boldsymbol{\theta}_j^{g,*} \right) + \left(\mathcal{U}^g_{j,1}\right)^\top \mathbf{1}
\right\|_2\right\} 
\leq 
n^g\lambda_0
\right\} \cap
\\
&\left\{
\left\| 
\left(\mathcal{V}_{j,k,2}^g\right)^\top \left(\mathcal{V}_{j,1}^g \boldsymbol{\alpha}_j^{g,*} + \mathcal{V}_{j,2}^g\boldsymbol{\theta}_j^{g,*}  \right) + \left(\mathcal{U}^g_{j,2}\right)^\top \mathbf{1}
\right\|_2
\leq 
n^g\lambda_0
\right\}
\end{align*}
for some $\lambda_0 \leq \lambda/2$.
Suppose the compatibility condition also holds.
Then on the set $\mathcal{E}$,
\begin{align*}
\mathcal{P}\left( \tilde{\boldsymbol{\alpha}}^g_j - \boldsymbol{\alpha}^{g,*}_j \right)
+
\| \tilde{\boldsymbol{\theta}}^{g}_j - \boldsymbol{\theta}_j^{g,*} \|_2
\leq \frac{\lambda 4 |S|}{C^2} .
\end{align*}
\end{theorem}

\begin{proof}[\textbf{Proof of Theorem 2}]
\noindent The regularized score matching estimator $\tilde{\boldsymbol{\alpha}}_j^{g}$ necessarily satisfies the following basic inequality:
\begin{align*}
L^g_{n,j}\left( \tilde{\boldsymbol{\alpha}}^g_j, \tilde{\boldsymbol{\theta}}^g_j\right) + \lambda\mathcal{P}_j\left( \tilde{\boldsymbol{\alpha}}^g_j \right) \leq
L^g_{n,j}\left( \boldsymbol{\alpha}^{g,*}_j, \boldsymbol{\theta}^{g,*}_j\right) + \lambda\mathcal{P}_j\left( \boldsymbol{\alpha}^{g,*}_j \right).
\end{align*}
With some algebra, this inequality can be rewritten as
\begin{align*}
&\left(2n^g\right)^{-1}
\begin{pmatrix}
\left(\tilde{\boldsymbol{\alpha}}^g_j - \boldsymbol{\alpha}^{g,*}_j \right)^\top
&
\left(\tilde{\boldsymbol{\theta}}^g_j - \boldsymbol{\theta}^{g,*}_j\right)^\top
\end{pmatrix}
\begin{pmatrix}
\left(\mathcal{V}_{j,1}^g\right)^\top\mathcal{V}_{j,1}^g  & \left(\mathcal{V}_{j,1}^g\right)^\top\mathcal{V}_{j,2}^g
\\
\left(\mathcal{V}_{j,2}^g\right)^\top\mathcal{V}_{j,1}^g & \left(\mathcal{V}_{j,2}^g\right)^\top\mathcal{V}_{j,2}^g
\end{pmatrix}
\begin{pmatrix}
\tilde{\boldsymbol{\alpha}}^g_j - \boldsymbol{\alpha}^{g,*}_j
\\
\tilde{\boldsymbol{\theta}}^g_j - \boldsymbol{\theta}^{g,*}_j
\end{pmatrix}
+
\lambda\mathcal{P}_j\left( \tilde{\boldsymbol{\alpha}}^g_j \right) 
\leq
\\
&
-\left(n^g\right)^{-1}
\begin{pmatrix}
\left(\tilde{\boldsymbol{\alpha}}^g_j - \boldsymbol{\alpha}^{g,*}_j \right)^\top
&
\left(\tilde{\boldsymbol{\theta}}^g_j - \boldsymbol{\theta}^{g,*}_j\right)^\top
\end{pmatrix}
\begin{pmatrix}
\left(\mathcal{V}_{j,1}^g\right)^\top \left(\mathcal{V}_{j,1}^g \boldsymbol{\alpha}_j^{g,*} + \mathcal{V}_{j,2}^g\boldsymbol{\theta}_j^{g,*} \right) + \left(\mathcal{U}^g_{j,1}\right)^\top \mathbf{1}
\\
\left(\mathcal{V}_{j,2}^g\right)^\top \left(\mathcal{V}_{j,1}^g \boldsymbol{\alpha}_j^{g,*} + \mathcal{V}_{j,2}^g\boldsymbol{\theta}_j^{g,*}  \right)  +  \left(\mathcal{U}^g_{j,2}\right)^\top \mathbf{1}
\end{pmatrix}
+
\lambda\mathcal{P}_j\left( \boldsymbol{\alpha}^{g,*}_j \right).
\end{align*}

\noindent By Lemma 1, on the set $\mathcal{E}$ and using $\lambda \geq \lambda_0/2$ we get
\begin{align*}
&
\left(n^g\right)^{-1}
\begin{pmatrix}
\left(\tilde{\boldsymbol{\alpha}}^g_j - \boldsymbol{\alpha}^{g,*}_j \right)^\top
&
\left(\tilde{\boldsymbol{\theta}}^g_j - \boldsymbol{\theta}^{g,*}_j\right)^\top
\end{pmatrix}
\begin{pmatrix}
\left(\mathcal{V}_{j,1}^g\right)^\top\mathcal{V}_{j,1}^g  & \left(\mathcal{V}_{j,1}^g\right)^\top\mathcal{V}_{j,2}^g
\\
\left(\mathcal{V}_{j,2}^g\right)^\top\mathcal{V}_{j,1}^g & \left(\mathcal{V}_{j,2}^g\right)^\top\mathcal{V}_{j,2}^g
\end{pmatrix}
\begin{pmatrix}
\tilde{\boldsymbol{\alpha}}^g_j - \boldsymbol{\alpha}^{g,*}_j
\\
\tilde{\boldsymbol{\theta}}^g_j - \boldsymbol{\theta}^{g,*}_j
\end{pmatrix}
+
2\lambda \mathcal{P}_j\left( \tilde{\boldsymbol{\alpha}}^g_j \right) 
\leq
\\
&
\lambda\left\|\tilde{\boldsymbol{\theta}}_j - \boldsymbol{\theta}^*_j \right\|_2
+
2\lambda \mathcal{P}_j\left( \boldsymbol{\alpha}^{g,*}_j \right) + \lambda\mathcal{P}_j\left( \tilde{\boldsymbol{\alpha}}^g_j - \boldsymbol{\alpha}^{g,*}_j \right).
\end{align*}
On the left hand side, we apply the triangle inequality to get
\begin{align*}
\mathcal{P}_j\left( \tilde{\boldsymbol{\alpha}}^g_j \right) = 
\mathcal{P}_j\left( \tilde{\boldsymbol{\alpha}}^g_j \circ \mathds{1}_{S} \right) 
+ 
\mathcal{P}_j\left( \tilde{\boldsymbol{\alpha}}^g_j \circ \mathds{1}_{\bar{S}} \right)
\geq
\mathcal{P}_j\left( \boldsymbol{\alpha}_j^{g,*}  \circ \mathds{1}_{S} \right)
-
\mathcal{P}_j\left(\left( \tilde{\boldsymbol{\alpha}}^g_j - \boldsymbol{\alpha}_j^{g,*} \right) \circ \mathds{1}_{S} \right)
+
\mathcal{P}_j\left( \tilde{\boldsymbol{\alpha}}^g_j \circ \mathds{1}_{\bar{S}} \right).
\end{align*}
On the right hand side, we observe that
\begin{align*}
\mathcal{P}_j\left( \tilde{\boldsymbol{\alpha}}^g_j - \boldsymbol{\alpha}^{g,*}_j \right) =
\mathcal{P}_j\left(\left( \tilde{\boldsymbol{\alpha}}^g_j - \boldsymbol{\alpha}_j^{g,*} \right) \circ \mathds{1}_{S} \right)
+
\mathcal{P}_j\left(\tilde{\boldsymbol{\alpha}}^g_j  \circ \mathds{1}_{\bar{S}} \right).
\end{align*}
We then have
\begin{align*}
&
\left(n^g\right)^{-1}
\begin{pmatrix}
\left(\tilde{\boldsymbol{\alpha}}^g_j - \boldsymbol{\alpha}^{g,*}_j \right)^\top
&
\left(\tilde{\boldsymbol{\theta}}^g_j - \boldsymbol{\theta}^{g,*}_j\right)^\top
\end{pmatrix}
\begin{pmatrix}
\left(\mathcal{V}_{j,1}^g\right)^\top\mathcal{V}_{j,1}^g  & \left(\mathcal{V}_{j,1}^g\right)^\top\mathcal{V}_{j,2}^g
\\
\left(\mathcal{V}_{j,2}^g\right)^\top\mathcal{V}_{j,1}^g & \left(\mathcal{V}_{j,2}^g\right)^\top\mathcal{V}_{j,2}^g
\end{pmatrix}
\begin{pmatrix}
\tilde{\boldsymbol{\alpha}}^g_j - \boldsymbol{\alpha}^{g,*}_j
\\
\tilde{\boldsymbol{\theta}}^g_j - \boldsymbol{\theta}^{g,*}_j
\end{pmatrix}
+
\lambda \mathcal{P}_j\left( \tilde{\boldsymbol{\alpha}}^g_j \circ \mathds{1}_{\bar{S}} \right) 
\leq
\\
&
\lambda\left\|\tilde{\boldsymbol{\theta}}^g_j - \boldsymbol{\theta}^{g,*}_j \right\|_2
+
3\lambda\mathcal{P}_j\left(\left( \tilde{\boldsymbol{\alpha}}^g_j - \boldsymbol{\alpha}_j^{g,*} \right) \circ \mathds{1}_{S} \right).
\end{align*}
Now, 
\begin{align*}
&
\left(n^g\right)^{-1}
\begin{pmatrix}
\left(\tilde{\boldsymbol{\alpha}}^g_j - \boldsymbol{\alpha}^{g,*}_j \right)^\top
&
\left(\tilde{\boldsymbol{\theta}}^g_j - \boldsymbol{\theta}^{g,*}_j\right)^\top
\end{pmatrix}
\begin{pmatrix}
\left(\mathcal{V}_{j,1}^g\right)^\top\mathcal{V}_{j,1}^g  & \left(\mathcal{V}_{j,1}^g\right)^\top\mathcal{V}_{j,2}^g
\\
\left(\mathcal{V}_{j,2}^g\right)^\top\mathcal{V}_{j,1}^g & \left(\mathcal{V}_{j,2}^g\right)^\top\mathcal{V}_{j,2}^g
\end{pmatrix}
\begin{pmatrix}
\tilde{\boldsymbol{\alpha}}^g_j - \boldsymbol{\alpha}^{g,*}_j
\\
\tilde{\boldsymbol{\theta}}^g_j - \boldsymbol{\theta}^{g,*}_j
\end{pmatrix}
+
\\
&\lambda \mathcal{P}_j\left( \tilde{\boldsymbol{\alpha}}^g_j - \boldsymbol{\alpha}_j^{g,*} \right) +
 \lambda\left\|\tilde{\boldsymbol{\theta}}^g_j - \boldsymbol{\theta}^{g,*}_j \right\|_2
=
\\ \\
&
\left(n^g\right)^{-1}\begin{pmatrix}
\left(\tilde{\boldsymbol{\alpha}}^g_j - \boldsymbol{\alpha}^{g,*}_j \right)^\top
&
\left(\tilde{\boldsymbol{\theta}}^g_j - \boldsymbol{\theta}^{g,*}_j\right)^\top
\end{pmatrix}
\begin{pmatrix}
\left(\mathcal{V}_{j,1}^g\right)^\top\mathcal{V}_{j,1}^g  & \left(\mathcal{V}_{j,1}^g\right)^\top\mathcal{V}_{j,2}^g
\\
\left(\mathcal{V}_{j,2}^g\right)^\top\mathcal{V}_{j,1}^g & \left(\mathcal{V}_{j,2}^g\right)^\top\mathcal{V}_{j,2}^g
\end{pmatrix}
\begin{pmatrix}
\tilde{\boldsymbol{\alpha}}^g_j - \boldsymbol{\alpha}^{g,*}_j
\\
\tilde{\boldsymbol{\theta}}^g_j - \boldsymbol{\theta}^{g,*}_j
\end{pmatrix}
+
\\
&\lambda \mathcal{P}_j\left( \tilde{\boldsymbol{\alpha}}^g_j \circ \mathds{1}_{\bar{S}} \right) + 
\lambda \mathcal{P}_j\left( \left( \tilde{\boldsymbol{\alpha}}^g_j - \boldsymbol{\alpha}_j^{g,*} \right) \circ \mathds{1}_{S} \right)
 +
 \lambda\left\|\tilde{\boldsymbol{\theta}}^g_j - \boldsymbol{\theta}^{g,*}_j \right\|_2
 \leq
\\ \\
&
\frac{\lambda 2|S|^{1/2}}{C} 
\left\{
\left(n^g\right)^{-1}
\begin{pmatrix}
\left(\tilde{\boldsymbol{\alpha}}^g_j - \boldsymbol{\alpha}^{g,*}_j \right)^\top
&
\left(\tilde{\boldsymbol{\theta}}^g_j - \boldsymbol{\theta}^{g,*}_j\right)^\top
\end{pmatrix}
\begin{pmatrix}
\left(\mathcal{V}_{j,1}^g\right)^\top\mathcal{V}_{j,1}^g  & \left(\mathcal{V}_{j,1}^g\right)^\top\mathcal{V}_{j,2}^g
\\
\left(\mathcal{V}_{j,2}^g\right)^\top\mathcal{V}_{j,1}^g & \left(\mathcal{V}_{j,2}^g\right)^\top\mathcal{V}_{j,2}^g
\end{pmatrix}
\begin{pmatrix}
\tilde{\boldsymbol{\alpha}}^g_j - \boldsymbol{\alpha}_j^{g,*} 
\\
\tilde{\boldsymbol{\theta}}^g_j - \boldsymbol{\theta}^{g,*}_j
\end{pmatrix}
\right\}^{1/2} \leq 
\\ \\
&\frac{\lambda^2 4 |S|}{C^2} 
+
\left(n^g\right)^{-1}
\begin{pmatrix}
\left(\tilde{\boldsymbol{\alpha}}^g_j - \boldsymbol{\alpha}^{g,*}_j \right)^\top
&
\left(\tilde{\boldsymbol{\theta}}^g_j - \boldsymbol{\theta}^{g,*}_j\right)^\top
\end{pmatrix}
\begin{pmatrix}
\left(\mathcal{V}_{j,1}^g\right)^\top\mathcal{V}_{j,1}^g  & \left(\mathcal{V}_{j,1}^g\right)^\top\mathcal{V}_{j,2}^g
\\
\left(\mathcal{V}_{j,2}^g\right)^\top\mathcal{V}_{j,1}^g & \left(\mathcal{V}_{j,2}^g\right)^\top\mathcal{V}_{j,2}^g
\end{pmatrix}
\begin{pmatrix}
\tilde{\boldsymbol{\alpha}}^g_j - \boldsymbol{\alpha}_j^{g,*} 
\\
\tilde{\boldsymbol{\theta}}^g_j - \boldsymbol{\theta}^{g,*}_j
\end{pmatrix},
\end{align*}
where we use the compatiblility condition for the first inequality, and for the second inequality use the fact that
\begin{align*}
ab \leq b^2 + a^2
\end{align*}
for any $a,b\in\mathds{R}$.  The conclusion follows immediately.
\end{proof}

If the event $\mathcal{E}$ occurs with probability tending to one, Theorem 2 implies
\begin{align*}
\mathcal{P}\left( \tilde{\boldsymbol{\alpha}}^g_j - \boldsymbol{\alpha}^{g,*}_j \right)
+
\| \tilde{\boldsymbol{\theta}}^{g}_j - \boldsymbol{\theta}_j^{g,*} \|_2
=
O_P\left(\lambda\right).
\end{align*}
We select $\lambda$ so that the event $\mathcal{E}$ occurs with high probability.
For instance, suppose the elements of the matrix
\begin{align*}
\xi = 
\begin{pmatrix}
\text{diag}\left(\mathcal{V}_{j,1}^g\boldsymbol{\alpha}^{g,*}_j  + \mathcal{V}_{j,2}^g \boldsymbol{\theta}^{g,*}_j \right)\mathcal{V}_{j,1}^g
\\
\text{diag}\left(\mathcal{V}_{j,1}^g\boldsymbol{\alpha}^{g,*}_j  + \mathcal{V}_{j,2}^g \boldsymbol{\theta}^{g,*}_j \right) \mathcal{V}_{j,2}^g
\end{pmatrix}
\end{align*}
are sub-Gaussian, and consider the event
\begin{align*}
\bar{\mathcal{E}} = 
&\left\{
\left\| 
\begin{pmatrix}
\left(\mathcal{V}_{j,1}^g\right)^\top \left(\mathcal{V}_{j,1}^g \boldsymbol{\alpha}_j^{g,*} + \mathcal{V}_{j,2}^g\boldsymbol{\theta}_j^{g,*}  \right) + \left(\mathcal{U}^g_{j,1}\right)^\top \mathbf{1}
\\
\left(\mathcal{V}_{j,2}^g\right)^\top \left(\mathcal{V}_{j,1}^g \boldsymbol{\alpha}_j^{g,*} + \mathcal{V}_{j,2}^g\boldsymbol{\theta}_j^{g,*}  \right)  +  \left(\mathcal{U}^g_{j,2}\right)^\top \mathbf{1}
\end{pmatrix}
\right\|_{\infty}
\leq 
\frac{n^g\lambda_0}{d}
\right\},
\end{align*}
where $\|\cdot\|_{\infty}$ is the $\ell_\infty$ norm.
Observing that $\mathcal{E} \subset \bar{\mathcal{E}}$, it is only necessary to show that $\bar{\mathcal{E}}$ holds with high probability.
It is shown in Corollary 2 of Negahban et al. \cite{negahban2012unified} that there exist constants $u_1, u_2 >0$ such that with $\lambda_0 \asymp \{\log(p)/n\}^{1/2}$, $\bar{\mathcal{E}}$ holds with probability at least $1 - u_1p^{-u_2}$.
Thus, $\mathcal{E}$ occurs with probability tending to one as $p \to \infty$.
For distributions with heavier tails, a larger choice of $\lambda$ may be  required \citep{yu2019generalized}.

\subsubsection*{B.4 \quad De-biased Score Matching Estimator}

The KKT conditions for the regularized score matching loss imply that the estimator $\tilde{\boldsymbol{\alpha}}^g_j$ satisfies
\begin{align*}
\nabla L_{n,j}(\tilde{\boldsymbol{\alpha}}^g_j, \tilde{\boldsymbol{\theta}}^g_j) &= 
\left(n^g\right)^{-1}
\begin{pmatrix}
\left(\mathcal{V}_{j,1}^g\right)^\top\mathcal{V}_{j,1}^g  & \left(\mathcal{V}_{j,1}^g\right)^\top\mathcal{V}_{j,2}^g
\\
\left(\mathcal{V}_{j,2}^g\right)^\top\mathcal{V}_{j,1}^g & \left(\mathcal{V}_{j,2}^g\right)^\top\mathcal{V}_{j,2}^g
\end{pmatrix}
\begin{pmatrix}
\tilde{\boldsymbol{\alpha}}_j^g
\\
\tilde{\boldsymbol{\theta}}_j^g
\end{pmatrix}
+
\left(n^g\right)^{-1}
\begin{pmatrix}
\left(\mathcal{U}_{j,1}^g\right)^\top\mathbf{1}
\\
\left(\mathcal{U}_{j,2}^g\right)^\top\mathbf{1}
\end{pmatrix}
=
\begin{pmatrix}
\lambda \nabla P\left( \tilde{\boldsymbol{\alpha}}^g_j \right)
\\
\mathbf{0}
\end{pmatrix}.
\end{align*}
With some algebra, we can rewrite the KKT conditions as
\begin{align*}
&\left(n^g\right)^{-1}
\begin{pmatrix}
\left(\mathcal{V}_{j,1}^g\right)^\top\mathcal{V}_{j,1}^g  & \left(\mathcal{V}_{j,1}^g\right)^\top\mathcal{V}_{j,2}^g
\\
\left(\mathcal{V}_{j,2}^g\right)^\top\mathcal{V}_{j,1}^g & \left(\mathcal{V}_{j,2}^g\right)^\top\mathcal{V}_{j,2}^g
\end{pmatrix}
\begin{pmatrix}
\tilde{\boldsymbol{\alpha}}^g_j - \boldsymbol{\alpha}_j^{g,*}
\\
\tilde{\boldsymbol{\theta}}^g_j - \boldsymbol{\theta}_j^{g,*}
\end{pmatrix}
=
\\
&\lambda
\begin{pmatrix}
\nabla P\left( \tilde{\boldsymbol{\alpha}}^g_j \right)
\\
\mathbf{0}
\end{pmatrix}
-
\left(n^g\right)^{-1}
\begin{pmatrix}
\left(\mathcal{V}_{j,1}^g\right)^\top \left(\mathcal{V}_{j,1}^g \boldsymbol{\alpha}_j^{g,*} + \mathcal{V}_{j,2}^g\boldsymbol{\theta}_j^{g,*}   \right) + \left(\mathcal{U}^g_{j,1}\right)^\top \mathbf{1}
\\
\left(\mathcal{V}_{j,2}^g\right)^\top \left(\mathcal{V}_{j,1}^g \boldsymbol{\alpha}_j^{g,*} + \mathcal{V}_{j,2}^g\boldsymbol{\theta}_j^{g,*}  \right)  +  \left(\mathcal{U}^g_{j,2}\right)^\top \mathbf{1}
\end{pmatrix}.
\end{align*}
Now, let $\Sigma_{j,n}$ be the matrix
\begin{align*}
\Sigma_{j,n} = 
\left(n^g\right)^{-1}
\begin{pmatrix}
\left(\mathcal{V}_{j,1}^g\right)^\top\mathcal{V}_{j,1}^g  & \left(\mathcal{V}_{j,1}^g\right)^\top\mathcal{V}_{j,2}^g
\\
\left(\mathcal{V}_{j,2}^g\right)^\top\mathcal{V}_{j,1}^g & \left(\mathcal{V}_{j,2}^g\right)^\top\mathcal{V}_{j,2}^g
\end{pmatrix},
\end{align*}
let $\Sigma_j = \mathds{E}[\Sigma_{j,n}]$, and let $\tilde{M}_j$ be an estimate of $\Sigma_j^{-1}$.
We can now rewrite the KKT conditions as
\begin{align}
\begin{pmatrix}
\tilde{\boldsymbol{\alpha}}^g_j - \boldsymbol{\alpha}_j^{g,*}
\\
\tilde{\boldsymbol{\theta}}^g_j - \boldsymbol{\theta}_j^{g,*}
\end{pmatrix}
&=
\underset{(\mathrm{i})}{\underbrace{\lambda \tilde{M}_j
\begin{pmatrix}
\nabla P\left( \tilde{\boldsymbol{\alpha}}^g_j \right)
\\
\mathbf{0}
\end{pmatrix}}}
-
\underset{(\mathrm{ii})}{\underbrace{\left(n^g\right)^{-1}
\tilde{M}_j
\begin{pmatrix}
\left(\mathcal{V}_{j,1}^g\right)^\top \left(\mathcal{V}_{j,1}^g \boldsymbol{\alpha}_j^{g,*} + \mathcal{V}_{j,2}^g\boldsymbol{\theta}_j^{g,*}  \right) + \left(\mathcal{U}^g_{j,1}\right)^\top \mathbf{1}
\\
\left(\mathcal{V}_{j,2}^g\right)^\top \left(\mathcal{V}_{j,1}^g \boldsymbol{\alpha}_j^{g,*} + \mathcal{V}_{j,2}^g\boldsymbol{\theta}_j^{g,*}   \right)  +  \left(\mathcal{U}^g_{j,2}\right)^\top \mathbf{1}
\end{pmatrix} \nonumber
}}
 +
\\
&\quad\quad\quad
\underset{(\mathrm{iii})}{
\underbrace{\left(n^g\right)^{-1}
\left\{
I - 
\Sigma_{j,n}
\tilde{M}_j
\right\}
\begin{pmatrix}
\tilde{\boldsymbol{\alpha}}^g_j - \boldsymbol{\alpha}_j^{g,*}
\\
\tilde{\boldsymbol{\theta}}^g_j - \boldsymbol{\theta}_j^{g,*}
\end{pmatrix}
}}.
\label{DebiasSMExp}
\end{align}
As is the case for the de-biased group LASSO in Appendix A, the first term $(\mathrm{i})$ in \eqref{DebiasSMExp} depends only on the observed data and can be directly subtracted from the initial estimate.  
The second term $(\mathrm{ii})$ is asymptotically equivalent to 
\begin{align}
\left(n^g\right)^{-1}\Sigma_j^{-1}
\begin{pmatrix}
\left(\mathcal{V}_{j,1}^g\right)^\top \left(\mathcal{V}_{j,1}^g \boldsymbol{\alpha}_j^{g,*} + \mathcal{V}_{j,2}^g\boldsymbol{\theta}_j^{g,*}  \right) + \left(\mathcal{U}^g_{j,1}\right)^\top \mathbf{1}
\\
\left(\mathcal{V}_{j,2}^g\right)^\top \left(\mathcal{V}_{j,1}^g \boldsymbol{\alpha}_j^{g,*} + \mathcal{V}_{j,2}^g\boldsymbol{\theta}_j^{g,*} \right)  +  \left(\mathcal{U}^g_{j,2}\right)^\top \mathbf{1}
\end{pmatrix},
\label{SMSlutsky}
\end{align}
if $\tilde{M}_j$ is a consistent estimate of $\Sigma_j^{-1}$.  Using the fact that $\mathds{E}\left[ \nabla L^g_{n,j}\left(\boldsymbol{\alpha}^{g,*}_j, \boldsymbol{\theta}^{g,*}_j  \right) \right] = \mathbf{0}$, it can be seen that \eqref{SMSlutsky} is an average of \textit{i.i.d.} random quantities with mean zero.
The central theorem then implies that any low-dimensional sub-vector is asymptotically normal.
The last term $(\mathrm{iii})$ is asymptotically negligible if $\tilde{M}_{j}$ is an approximate inverse of $\Sigma_{j,n}$ and if $(\tilde{\boldsymbol{\alpha}}_j^g, \tilde{\boldsymbol{\theta}}_j^g)$ is consistent for $(\boldsymbol{\alpha}_j^{g,*}, \boldsymbol{\theta}_j^{g,*})$.
Thus, for an appropriate choice of $\tilde{M}_j$, we expect asymptotic normality of an estimator of the form
\begin{align*}
\begin{pmatrix}
\check{\boldsymbol{\alpha}}^g_j
\\
\check{\boldsymbol{\theta}}^g_j
\end{pmatrix}
=
\begin{pmatrix}
\tilde{\boldsymbol{\alpha}}^g_j
\\
\tilde{\boldsymbol{\theta}}^g_j
\end{pmatrix} -
\lambda \tilde{M}_j
\begin{pmatrix}
\nabla P\left( \tilde{\boldsymbol{\alpha}}^g_j \right)
\\
\mathbf{0}
\end{pmatrix}.
\end{align*}

Before constructing $\tilde{M}_j$, we first provide an alternative expression for $\Sigma_j^{-1}$.
We define the $d \times d$ matrices $\Gamma^*_{j,k,l}$ and $\Delta^*_{j,k}$ as
\begin{align*}
&\Gamma^*_{j,k,1},\ldots,\Gamma^*_{j,k,p}, \Delta^*_{j,k}  =  \nonumber
\\
&\underset{\Gamma_1,\ldots,\Gamma_p, \Delta \in \mathds{R}^{d\times d}}{\text{arg min}}
\mathds{E}\left[
\text{trace}\left\{ 
\left(n^g\right)^{-1}
\left(\mathcal{V}_{j,k,1}^g - \sum_{l \neq k,j} \mathcal{V}_{j,l,1}^g \Gamma_l - \mathcal{V}^g_{j,2}\Delta \right)^\top 
\left(\mathcal{V}_{j,k,1}^g - \sum_{l \neq k,j} \mathcal{V}_{j,l,1}^g \Gamma_l - \mathcal{V}^g_{j,2}\Delta \right)
\right\}\right].
\end{align*}
We also define the $d \times d$ matrices $\Lambda^*_{j,k}$ as
\begin{align*}
&\Lambda^*_{j,1},\ldots,\Lambda^*_{j,p}  =  \nonumber
\underset{\Lambda_1,\ldots,\Lambda_p, \in \mathds{R}^{d\times d}}{\text{arg min}}
\mathds{E}\left[
\text{trace}\left\{ \left(n^g\right)^{-1}
\left(\mathcal{V}_{j,2}^g - \sum_{k \neq j} \mathcal{V}_{j,k,1}^g \Lambda_{k} \right)^\top 
\left(\mathcal{V}_{j,2}^g - \sum_{k \neq j} \mathcal{V}_{j,k,1}^g \Lambda_{k} \right)
\right\}\right].
\end{align*}
Additionally, we define the $d \times d$ matrices $C^*_{j,k}$ and $D^*_j$
\begin{align*}
&C^*_{j,k} = \mathds{E}\left[\left(n^g\right)^{-1}\left(\mathcal{V}^g_{j,k,1}\right)^\top
\left(\mathcal{V}_{j,k,1}^g - \sum_{l \neq k,j} \mathcal{V}_{j,l,1}^g \Gamma^*_{j,k,l} - \mathcal{V}^g_{j,2}\Delta^*_{j,k} \right)\right]
 \\
&D^*_j = \mathds{E}\left[\left(n^g\right)^{-1}\left(\mathcal{V}^g_{j,2}\right)^\top
\left(\mathcal{V}_{j,2}^g - \sum_{k \neq j} \mathcal{V}_{j,k,1}^g \Lambda^*_{j,k} \right)\right].
\end{align*}
It can be shown that $\Sigma_j^{-1}$ can be expressed as
\begin{align*}
\Sigma^{-1}_j =
\begin{pmatrix}
\left(C^*_{j,1}\right)^{-1} & \cdots & \mathbf{0} & \mathbf{0}
\\
\vdots & \ddots & \vdots & \vdots
\\
\mathbf{0} & \cdots & \left(C^*_{j,p}\right)^{-1} & \mathbf{0} 
\\
\mathbf{0} & \cdots & \mathbf{0} & \left(D^*_{j}\right)^{-1}
\end{pmatrix}
\begin{pmatrix}
I & -\Gamma^*_{j,1,2} & \cdots & -\Gamma^*_{j,1,p} & - \Delta^*_{j,1} 
\\
-\Gamma^*_{j,2,1} & I & \cdots & -\Gamma^*_{j,2,p}  & - \Delta^*_{j,2}
\\
\vdots & \vdots & \ddots & \vdots & \vdots
\\
-\Gamma^*_{j,p,1} & -\Gamma^*_{j,p,2} & \cdots & I & - \Delta^*_{j,p}
\\
-\Lambda^*_{j,1} & -\Lambda^*_{j,2} & \cdots & -\Lambda^*_{j,p} & I
\end{pmatrix}
.
\end{align*}
We can thus estimate $\Sigma_j^{-1}$ by estimating each of the matrices $\Gamma^*_{j,k,l}$, $\Lambda^*_{j,k}$, and $\Delta^*_{j,k}$.

Similar to our discussion of the de-biased group LASSO in Appendix A, we use a group-penalized variant of the nodewise LASSO to construct $\tilde{M}_j$.
We estimate $\Gamma^*_{j,k,l}$ and $\Delta^*_{j,k}$ as
\begin{align*}
&\tilde{\Gamma}_{j,k,1},\ldots,\tilde{\Gamma}_{j,k,p}, \tilde{\Delta}_{j,k}  =  \nonumber
\\
&\underset{\Gamma_1,\ldots,\Gamma_p, \Delta \in \mathds{R}^{d\times d}}{\text{arg min}}
\text{trace}\left\{ 
\left(n^g\right)^{-1}
\left(\mathcal{V}_{j,k,1}^g - \sum_{l \neq k,j} \mathcal{V}_{j,l,1}^g \Gamma_l - \mathcal{V}^g_{j,2}\Delta \right)^\top 
\left(\mathcal{V}_{j,k,1}^g - \sum_{l \neq k,j} \mathcal{V}_{j,l,1}^g \Gamma_l - \mathcal{V}^g_{j,2}\Delta \right)
\right\}
 +
 \\
&\quad\quad\quad\quad\quad\omega_1 \sum_{l \neq k,j} \|\Gamma_l \|_{2,*} + \omega_1 \|\Delta\|_{2,*},
\end{align*}
where $\omega_1,\omega_2 > 0$ are tuning parameters, and $\|\cdot\|_{2,*}$ is as defined in Appendix A.
We estimate $\Lambda^*_{j,k}$ as
\begin{align}
&\tilde{\Lambda}_{j,1},\ldots,\tilde{\Lambda}_{j,p}  =  \nonumber
\underset{\Lambda_1,\ldots,\Lambda_p, \in \mathds{R}^{d\times d}}{\text{arg min}}
\text{trace}\left\{ \left(n^g\right)^{-1}
\left(\mathcal{V}_{j,2}^g - \sum_{k \neq j} \mathcal{V}_{j,k,1}^g \Lambda_{k} \right)^\top 
\left(\mathcal{V}_{j,2}^g - \sum_{k \neq j} \mathcal{V}_{j,k,1}^g \Lambda_{k} \right)
\right\} + \omega_2 \sum_{l \neq j} \|\Lambda_k\|_{2,*}.
\label{Ugh2}
\end{align}
Additionally, we define the $d \times d$ matrices $\tilde{C}_{j,k}$ and $\tilde{D}_j$
\begin{align*}
&\tilde{C}_{j,k} = \left(n^g\right)^{-1}\left(\mathcal{V}^g_{j,k,1}\right)^\top
\left(\mathcal{V}_{j,k,1}^g - \sum_{l \neq k,j} \mathcal{V}_{j,l,1}^g \tilde{\Gamma}_{j,k,l} - \mathcal{V}^g_{j,2}\tilde{\Delta}_{j,k} \right)
 \\
&\tilde{D}_j = \left(n^g\right)^{-1}\left(\mathcal{V}^g_{j,2}\right)^\top
\left(\mathcal{V}_{j,2}^g - \sum_{k \neq j} \mathcal{V}_{j,k,1}^g \tilde{\Lambda}_{j,k} \right).
\end{align*}
We then take $\tilde{M}_j$ as
\begin{align*}
\tilde{M}_j =
\begin{pmatrix}
\tilde{C}^{-1}_{j,1} & \cdots & \mathbf{0} & \mathbf{0}
\\
\vdots & \ddots & \vdots & \vdots
\\
\mathbf{0} & \cdots & \tilde{C}^{-1}_{j,p} & \mathbf{0} 
\\
\mathbf{0} & \cdots & \mathbf{0} & \tilde{D}^{-1}_{j}
\end{pmatrix}
\begin{pmatrix}
I & -\tilde{\Gamma}_{j,1,2} & \cdots & -\tilde{\Gamma}_{j,1,p} & - \tilde{\Delta}_{j,1} 
\\
-\tilde{\Gamma}_{j,2,1} & I & \cdots & -\tilde{\Gamma}_{j,2,p}  & - \tilde{\Delta}_{j,2}
\\
\vdots & \vdots & \ddots & \vdots & \vdots
\\
-\tilde{\Gamma}_{j,p,1} & -\tilde{\Gamma}_{j,p,2} & \cdots & I & - \tilde{\Delta}_{j,p}
\\
-\tilde{\Lambda}_{j,1} & -\tilde{\Lambda}_{j,2} & \cdots & -\tilde{\Lambda}_{j,p} & I
\end{pmatrix}
.
\end{align*}

When $\Gamma^*_{j,k,l}$, $\Delta^*_{j,k}$, and $\Lambda^*_{j,k}$ satisfy appropriate sparsity conditions and some additional regularity assumptions, $\tilde{M}_j$ is a consistent estimate of $\Sigma_j^{-1}$ for $\omega_1 \asymp \{\log(p)/n\}^{1/2}$ and $\omega_2 \asymp \{\log(p)/n\}^{1/2}$ (see, e.g., Chapter 8 of B\"uhlmann and van de Geer \citep{buhlmann2011statistics} for a more comprehensive discussion).
Using the same argument presented in Appendix A, we are able to obtain the following bound on a scaled version of the remainder term $(\mathrm{iii})$:
\begin{align*}
&\left\|
\begin{pmatrix}
\tilde{C}_{j,1} & \cdots & \mathbf{0} & \mathbf{0}
\\
\vdots & \ddots & \vdots & \vdots
\\
\mathbf{0} & \cdots & \tilde{C}_{j,p} & \mathbf{0} 
\\
\mathbf{0} & \cdots & \mathbf{0} & \tilde{D}_{j}
\end{pmatrix}
\left\{
I - 
\left(n^g\right)^{-1}
\begin{pmatrix}
\left(\mathcal{V}_{j,1}^g\right)^\top\mathcal{V}_{j,1}^g  & \left(\mathcal{V}_{j,1}^g\right)^\top\mathcal{V}_{j,2}^g
\\
\left(\mathcal{V}_{j,2}^g\right)^\top\mathcal{V}_{j,1}^g & \left(\mathcal{V}_{j,2}^g\right)^\top\mathcal{V}_{j,2}^g
\end{pmatrix}
\tilde{M}_j
\right\}
\begin{pmatrix}
\tilde{\boldsymbol{\alpha}}^g_j - \boldsymbol{\alpha}_j^{g,*}
\\
\tilde{\boldsymbol{\theta}}^g_j - \boldsymbol{\theta}_j^{g,*}
\end{pmatrix} 
\right\|_{\infty} \leq 
\\
&\max\{\omega_1, \omega_2 \}
\left\{
\mathcal{P}\left( \tilde{\boldsymbol{\alpha}}^g_j - \boldsymbol{\alpha}^{g,*}_j \right)
+
\| \tilde{\boldsymbol{\theta}}^{g}_j - \boldsymbol{\theta}_j^{g,*} \|_2
\right\}.
\end{align*}
The remainder is $o_P(n^{-1/2})$ and hence asymptotically negligible if $n^{1/2} \max\{\omega_1, \omega_2\} \lambda \to 0$, where $\lambda$ is the tuning parameter for the regularized score matching estimator (see Theorem 2).

The de-biased estimate $\check{\alpha}^g_{j,k}$ of $\alpha^{g,*}_{j,k}$ can be expressed as
\begin{align}
\check{\alpha}^g_{j,k} = 
\tilde{\alpha}^g_{j,k} -
\left(n^g\right)^{-1}  
\tilde{C}^{-1}_{j,k}
\left(
 \mathcal{V}^g_{j,k,1} -
 \sum_{l \neq j, k} \mathcal{V}_{j,l,1}^g \tilde{\Gamma}_{j,k,l}
 \right)^\top
 \left( \mathcal{V}^g_{j,1} \tilde{\boldsymbol{\alpha}}^g_j + \mathcal{V}^g_{j,2} \tilde{\boldsymbol{\theta}}_j^g + \left(\mathcal{U}_{j,1}^g\right)^\top \mathbf{1} \right).
\end{align}
The difference between the de-biased estimator $\check{\alpha}^g_{j,k}$ and the true parameter $\alpha^{g,*}_{j,k}$ can be expressed as
\begin{align*}
\tilde{C}_{j,k}\left(\check{\alpha}^g_{j,k} - \alpha^{g,*}_{j,k}\right) =  
-&\left(n^g\right)^{-1}
\left(
 \mathcal{V}^g_{j,k,1} -
 \sum_{l \neq j, k} \mathcal{V}_{j,l,1}^g \Gamma^*_{j,k,l}
 \right)^\top
 \left( \mathcal{V}^g_{j,1} \boldsymbol{\alpha}^{g,*}_j + \mathcal{V}^g_{j,2} \boldsymbol{\theta}_j^{g,*} + \left(\mathcal{U}_{j,1}^g\right)^\top \mathbf{1} \right) +
\\
&\left(n^g\right)^{-1}
\left(\mathcal{V}^g_{j,2} \Delta^*_{j,k}\right)^\top
 \left( \mathcal{V}^g_{j,1} \boldsymbol{\alpha}^{g,*}_j + \mathcal{V}^g_{j,2} \boldsymbol{\theta}_j^{g,*}  + \left(\mathcal{U}_{j,2}^g\right)^\top \mathbf{1} \right)
\bigg\}
+ o_P\left(n^{-1/2}\right).
\end{align*}
As discussed above, the central limit theorem implies asymptotic normality of $\check{\alpha}^g_{j,k}$.
We can estimate the asymptotic variance of $\check{\alpha}^g_{j,k}$ as
\begin{align*}
\left(n^g\right)^{-2}\tilde{C}_{j,k}^{-1}\tilde{M}_{j,k}\tilde{\xi}^\top\tilde{\xi}\tilde{M}^\top_{j,k} \left(\tilde{C}_{j,k}^{-1}\right)^\top,
\end{align*}
where we define
\begin{align*}
\tilde{\xi} &=
\begin{pmatrix}
\text{diag}\left(\mathcal{V}_{j,1}^g \tilde{\boldsymbol{\alpha}}_j^{g} + \mathcal{V}_{j,2}^g\tilde{\boldsymbol{\theta}}_j^{g} \right)\mathcal{V}_{j,1}^g  + \mathcal{U}^g_{j,1}
\\
 \text{diag}\left(\mathcal{V}_{j,1}^g \tilde{\boldsymbol{\alpha}}_j^{g} + \mathcal{V}_{j,2}^g\tilde{\boldsymbol{\theta}}_j^{g}   \right)\mathcal{V}_{j,2}^g  +  \mathcal{U}^g_{j,2}
\end{pmatrix}
\\
\tilde{M}_{j,k} &= \begin{pmatrix}
-\tilde{\Gamma}_{j,k,1} & \cdots & -\tilde{\Gamma}_{j,k,k-1} & I & -\tilde{\Gamma}_{j,k,k+1} & \cdots & -\tilde{\Gamma}_{j,k,p} & - \tilde{\Delta}_{j,p}
\end{pmatrix}.
\end{align*}

\end{document}